\documentclass[12pt]{article}

\usepackage{amsmath}
\usepackage{amssymb}
\usepackage{amsthm}
\usepackage{hyperref}
\usepackage{enumerate}
\usepackage{bbm}
\usepackage{bm}
\usepackage{color}
\usepackage{chngpage}
\usepackage{graphicx}

\newtheorem{thm}{Theorem}[section]

\newtheorem{lem}[thm]{Lemma}

\newtheorem{assumption}[thm]{Assumption}

\newtheorem{definition}[thm]{Definition}

\newtheorem{example}[thm]{Example}
\newenvironment{exmp}{\begin{example}\rm}{\end{example}}
\newtheorem{remark}[thm]{Remark}
\newenvironment{rem}{\begin{remark}\rm}{\end{remark}}
\newtheorem{algorithm}[thm]{Algorithm}
\newenvironment{algo}{\begin{algorithm}\rm}{\end{algorithm}}

\title{Feedback Capacity of Stationary Gaussian Channels Further Examined~\footnote{Results in this paper have been partially presented in the 2017 IEEE ISIT~\cite{liu17}.}}

\author{\small \begin{tabular}{ccc}
Tao Liu & Guangyue Han\\
The University of Hong Kong & The University of Hong Kong\\
email:  liutao159003@gmail.com & email: ghan@hku.hk\\
\end{tabular}}

\date{\today}

\begin{document} \maketitle

\begin{abstract}
It is well known that the problem of computing the feedback capacity of a stationary Gaussian channel can be recast as an infinite-dimensional optimization problem; moreover, necessary and sufficient conditions for the optimality of a solution to this optimization problem have been characterized, and based on this characterization, an explicit formula for the feedback capacity has been given for the case that the noise is a first-order autoregressive moving-average Gaussian process. In this paper, we further examine the above-mentioned infinite-dimensional optimization problem. We prove that unless the Gaussian noise is white, its optimal solution is unique, and we propose an algorithm to recursively compute the unique optimal solution, which is guaranteed to converge in theory and features an efficient implementation for a suboptimal solution in practice. Furthermore, for the case that the noise is a $k$-th order autoregressive moving-average Gaussian process, we give a relatively more explicit formula for the feedback capacity; more specifically, the feedback capacity is expressed as a simple function evaluated at a solution to a system of polynomial equations, which is amenable to numerical computation for the cases $k=1, 2$ and possibly beyond.
\end{abstract}

\section{Introduction}

We consider the following additive Gaussian channel with feedback
\begin{equation} \label{gc}
Y_i=X_i(M, Y_1^{i-1})+Z_i, \quad i=1,2,\dots
\end{equation}
where $M$ denotes the message to be communicated through the channel, the noise $\{Z_i\}$, which is independent of $M$, is a zero mean stationary Gaussian process, and $X_i$, the channel input at time $i$, may depend on $M$ and previous channel outputs $Y_1^{i-1}$. And we assume the channel input $\{X_i\}$ satisfies the following average power constraint:
there is $P > 0$ such that for all $n$,
$$
\frac{1}{n} \sum_{i=1}^n E[(X_i(M, Y_1^{i-1}))^2] \le P.
$$
Let $C_{FB}$ denote the capacity of the channel (\ref{gc}), which is often referred to as {\em Gaussian feedback capacity} in the literature.

\par It is well known that the {\em non-feedback capacity} of (\ref{gc}) can be obtained through the power spectral density (PSD) water-filling method~\cite{shannon49}. As a matter of fact, when the channel noise is white (i.e., $\{Z_i\}$ is i.i.d.), Shannon~\cite{shannon56} showed that feedback does not increase capacity, which means, like its non-feedback counterpart, the feedback capacity features an explicit and simple formula (Here we note that in~\cite{kadota71a},~\cite{kadota71b}, Kadota, Zakai and Ziv also proved this statement for continuous-time white Gaussian channels). On the other hand though, if the channel is not white, feedback may increase capacity (see~\cite{ozarow90a},~\cite{ozarow90b}), and little has been known about its feedback capacity despite a number of papers~\cite{dembo89},~\cite{pinsker69},~\cite{ebert70},~\cite{cover89} relating the two capacities. Computing $C_{FB}$ has been a long-standing open problem that is of fundamental importance in information theory.

\par An prominent approach to tackle Gaussian feedback capacity can be found in a pioneering work~\cite{cover89}, where Cover and Pombra characterized the capacity through the sequence of the so-called ``$n$-block feedback capacity'':
\begin{equation}  \label{n-block-0}
C_{FB, n}=\max_{\mbox{tr}(K_{X,n})\le nP}\frac{1}{2n}\log\frac{\det(K_{Y,n})}{\det(K_{Z,n})},
\end{equation}
where $K_{X,n}$, $K_{Y,n}$, $K_{Z,n}$ stand for the covariance matrices of $X^n$, $Y^n$ and $Z^n$, respectively. It is also shown that the maximization can be taken over $X^n$ of the special form $X^n=B_nZ^n+V^n$, where $B_n$ is a strictly lower-triangular $n\times n$ matrix and the Gaussian vector $V^n$ is independent of $Z^n$. So, (\ref{n-block-0}) can be rewriten as
\begin{equation}  \label{n-block}
C_{FB, n}=\max_{B_n,K_{V,n}}\frac{1}{2n}\log\frac{\det((B_n+I)K_{Z,n}(B_n+I)^T+K_{V,n})}{\det(K_{Z,n})},
\end{equation}
subject to the constraint
$$
\mbox{tr}(B_nK_{Z,n}B_n^T+K_{V,n})\le nP,
$$
where $K_{V,n}$ is a negative semi-definite $n\times n$ matrix. Then, using the asymptotic equipartition property for arbitrary (non-stationary non-ergodic) Gaussian processes, a coding theorem can then be proved to characterize the Gaussian feedback capacity as the limiting expression below:
\begin{equation}\label{limitnblock}
  C_{FB}=\lim_{n \to \infty} C_{FB, n}.
\end{equation}
Though considerable efforts have been devoted to follow up the Cover-Pombra formulation, a ``computable'' formula for the Gaussian feedback capacity does not seem to be within sight: it is already difficult to find the sequence of the optimal $\{B_n,K_{V,n}\}$ acheiving $\{C_{FB, n}\}$, and its limiting behavior seems to be as evasive.

\par Another prominent approach came along in a recent work of Kim~\cite{kim10}, which led to a number of breakthroughs deepening our understanding of Gaussian feedback capacity. Roughly speaking, instead of examining the channel (\ref{gc}) over a finite time window, Kim justifies certain interchanges between limits and integrals when evaluating (\ref{n-block}) and (\ref{limitnblock}) and recast the problem of computing $C_{FB}$ as an infinite-dimensional optimization problem. Below, we state one of the theorems in~\cite{kim10} that is relevant to our results.

\begin{thm}[Theorem $4.1$ of~\cite{kim10}] \label{theorem1}
Suppose that the power spectral density $S_Z(e^{i\theta})$ of the Gaussian noise process $\{Z_i\}_{i=1}^\infty$ is bounded away from 0, and has a canonical spectral factorization $S_Z(e^{i\theta})=|H_Z(e^{i\theta})|^2$, where $H_Z(e^{i\theta}) \in \mathcal{H}_2$. Then the feedback capacity $C_{FB}$ is given by
\begin{equation} \label{starting-point1}
C_{FB}=\max_{B}\frac{1}{2}\int_{-\pi}^{\pi}\log|1+B(e^{i\theta})|^2 S_Z(e^{i\theta}) \frac{d\theta}{2\pi},
\end{equation}
where the maximum is taken over all strictly causal $B(e^{i\theta})$ satisfying the power constraint
$$\int_{-\pi}^{\pi}|B(e^{i\theta})|^2S_Z(e^{i\theta})\frac{d\theta}{2\pi}\le P.$$
Furthermore, a filter $B^{\star}(e^{i\theta})$ attains the maximum in (\ref{starting-point1}) if and only if
\begin{enumerate}
  \item[i)] Power:
  $$
  \int_{-\pi}^{\pi}|B^{\star}(e^{i\theta})|^2S_Z(e^{i\theta})\frac{d\theta}{2\pi}= P;       $$
  \item[ii)] Output spectrum:
  $$
  \eta:=\mathop{essinf}\limits_{\theta\in[-\pi,\pi)}|1+B^{\star}(e^{i\theta})|^2S_Z(e^{i\theta})>0;
  $$
  \item[iii)] Strong orthogonality: For some $0<\lambda\le\eta$
  \begin{equation} \label{right-version}
  \frac{\lambda}{1+B^{\star}(e^{i \theta})}-B^{\star}(e^{-i \theta})S_Z(e^{i \theta})
  \end{equation}
  is causal.
\end{enumerate}
\end{thm}

Using Theorem~\ref{theorem1} and relevant tools from the theory of Hardy spaces, Kim further characterized the capacity achieving $B(e^{i \theta})$ for the special case that $\{Z_i\}$ is a $k$-th order autoregressive moving-average (ARMA($k$)) Gaussian process. Roughly speaking, the following theorem says that the optimal $B$ must be rational satisfying three conditions corresponding to those in Theorem~\ref{theorem1}.

\begin{thm}[Proposition $5.1$ of~\cite{kim10}] \label{theorem2}
Suppose the noise $\{Z_i\}$ is not white and is an ARMA($k$) Gaussian process with parameters $\alpha_i, \beta_i$, $|\alpha_i| <1$, $|\beta_i|<1$ for all $i=1, 2, \dots, k$, namely, it has the power spectral density
{\small \begin{equation} \label{armaknoise}
S_Z(e^{i\theta})=|H_Z(e^{i\theta})|^2=\left|\frac{P(e^{i\theta})}{Q(e^{i\theta})}\right|^2=\left|\frac{\prod_{i=1}^{k}(1+\alpha_i e^{i\theta})}{\prod_{i=1}^{k}(1+\beta_i e^{i\theta})}\right|^2.
\end{equation}}
Then the feedback capacity $C_{FB}$ in (\ref{starting-point1}) is necessarily achieved by a filter $B$ of the form
\begin{equation}\label{filterb}
B(e^{i \theta})=b(e^{i \theta})\frac{R(e^{i \theta})}{P(e^{i \theta})}-1,
\end{equation}
where $R(z)$ is a stable polynomial whose degree is at most $k$, and
$$
b(z)=\frac{A(z)}{A^\#(z)}=\frac{\prod_{n}(1-\gamma_n^{-1}z)}{\prod_{n}(1-\gamma_n z)}
$$
is a normalized Blaschke product of at most $k$ zeros. Furthermore, a filter $B^{\star}(e^{i \theta})$ of the form (\ref{filterb}) is optimal if and only if the following hold:
  \begin{enumerate}
    \item[i)] Power:
    $$
    \int_{-\pi}^{\pi}|B^{\star}(e^{i\theta})|^2S_Z(e^{i\theta})\frac{d\theta}{2\pi}= P;
    $$
    \item[ii)] Output spectrum: For all zeros $\gamma_n$ of b(z)
    $$
    0<S_{Y}^{\star}(\gamma_n)=\lambda \le \min_{\theta\in[-\pi,\pi)}S_{Y}^{\star}(e^{i\theta});
    $$
    \item[iii)] Factorization:
    $$
    P(z)A^\#(z)-R(z)A(z)
    $$
    has a factor $Q(z)$.
  \end{enumerate}
\end{thm}

When applied to the case $k=1$, Theorem~\ref{theorem2} readily yields a rather tractable expression for the capacity achieving $B$ and gives a simple and explicit formula for $C_{FB}$, as detailed in the following theorem.
\begin{thm}[Theorem $5.3$ in~\cite{kim10}] \label{arma1-theorem}
Suppose the noise process $\{Z_i\}$ is an ARMA($1$) Gaussian process with parameters $\alpha$ and $\beta$, $|\alpha| < 1$, $|\beta| < 1$. Then, the Gaussian feedback capacity is given by
\begin{equation} \label{arma1-xx}
C_{FB}=-\frac{1}{2}\log x^2,
\end{equation}
where $x$ is the unique root of the following fourth-order polynomial
\begin{equation} \label{arma1-equation}
Px^2=\frac{(1-x^2)(1+\alpha x)^2}{(1+\beta x)^2},
\end{equation}
satisfying
\begin{equation}\label{conditionofx}
x \in \begin{cases}
(-1, 0) & \mbox{ if } \alpha \geq \beta,\\
(0, 1)  & \mbox{ if } \alpha < \beta.
\end{cases}
\end{equation}
\end{thm}

We now digress a bit to briefly mention related results on the ARMA($1$) Gaussian feedback capacity in the literature: Generalizing the celebrated Schalkwijk-Kailath scheme~\cite{sk66},~\cite{schalkwijk66}, Butman~\cite{butman69} obtained a lower bound of the feedback capacity of AR($1$) channel (a special ARMA($1$) channel with $\alpha=0$). Butman's bound was shown to be optimal under some cases of linear feedback schemes by Wolfowitz~\cite{wolfowitz75} and Tiernan~\cite{tiernan76}. Tiernan and Schalkwijk~\cite{ts74} also found an upper bound of AR($1$) Gaussian channel capacity, which is equal to Butman's lower bound at very low and very high signal-to-noise ratio. It was shown~\cite{kim06} that Butman's lower bound is indeed the capacity, and the capacity of MA($1$) channel (a special ARMA($1$) channel with $\beta=0$) was also derived in the same paper. More recently, Yang, Kav\v{c}i\'{c} and Tatikonda~\cite{yang07} studied the ARMA($k$) Gaussian channel by analyzing the structure of the optimal input distribution and reformulating the problem as a stochastic control optimization problem. And based on a speculation of the limiting behavior of the optimal input distribution, they derived the formula (\ref{arma1-xx}) and conjectured that it gives the ARMA($1$) Gaussian feedback capacity.

As mentioned above, the power of the variational formulation as in Theorem~\ref{theorem1} has been showcased in Theorem~\ref{arma1-theorem}, where the conjecture of~\cite{yang07} has been confirmed and the ARMA($1$) Gaussian feedback capacity is given as an explicit and simple formula. To the best of our knowledge, the ARMA($1$) Gaussian feedback channel is the only non-trivial scenario whose Gaussian feedback capacity is ``explicit''. The success by the variational formulation approach, contrasted by all the above-mentioned other approaches that have been struggling dealing with special cases of an ARMA($1$) channel, naturally posed the question of whether it can be extended to deal with more general channels, for instance, ARMA($k$) Gaussian feedback channels. Attempts in this direction, however, have somehow encountered certain technical barriers, due to the fact that the form in (\ref{filterb}) is ``less manageable'' (see Page $78$ in~\cite{kim10}). As a matter of fact, instead of following the variational formulation framework, an alternative state-space representation approach has been proposed in~\cite{kim10} to deal with the ARMA($k$) Gaussian feedback capacity, only to yield an intractable optimization problem  (see Theorem $6.1$ in~\cite{kim10}). Here we remark that prior to~\cite{kim10}, a result of similar nature has also been derived in Theorem $6$ of~\cite{yang07}, which however appears to be equally intractable.

In this paper, we will position ourselves within Kim's framework~\cite{kim10} and further examine feedback capacity of a stationary Gaussian channel as in (\ref{gc}). Our starting point is precisely Theorem~\ref{theorem1}, but instead of considering the filter $B(e^{i \theta})$, we use the method of ``change of variables'' and consider
\begin{equation} \label{c-def}
C(e^{i \theta}) \triangleq B(e^{i \theta})H_Z(e^{i \theta});
\end{equation}
here we note that since $B(e^{i \theta})$ is strictly causal and $H_Z(e^{i\theta}) \in \mathcal{H}_2$, it is obvious that $C(e^{i \theta})$ is also strictly causal, and thereby can be written as $C(e^{i\theta})=\sum_{k=1}^{\infty}c_ke^{ik\theta}$ for some $c_1, c_2, \dots \in \mathbb{R}$. Apparently, (\ref{c-def}) can be used to reformulate other quantities, such as the PSD of the channel output
\begin{equation} \label{output-PSD}
S_Y(e^{i \theta})=|C(e^{i \theta})+H(e^{i \theta})|^2,
\end{equation}
and eventually reformulate Theorem~\ref{theorem1} as follows:
\begin{thm}[Theorem $4.1$ of~\cite{kim10} reformulated] \label{theorem1-reformulated}
Suppose that the power spectral density $S_Z(e^{i\theta})$ of the Gaussian noise process $\{Z_i\}_{i=1}^\infty$ is bounded away from 0, and has a canonical spectral factorization $S_Z(e^{i\theta})=|H_Z(e^{i\theta})|^2$, where $H_Z(e^{i\theta})\in \mathcal{H}_2$. Then the feedback capacity $C_{FB}$ is given by
\begin{equation} \label{starting-point2}
C_{FB}=\max_{C}\frac{1}{2}\int_{-\pi}^{\pi}\log|C(e^{i\theta})+H(e^{i \theta})|^2 \frac{d\theta}{2\pi},
\end{equation}
where the maximum is taken over all strictly causal $C(e^{i\theta})$ satisfying the power constraint
\begin{equation}  \label{c-power}
\int_{-\pi}^{\pi}|C(e^{i\theta})|^2\frac{d\theta}{2\pi}\le P.
\end{equation}
Furthermore, a $C^{\star}(e^{i\theta})$ attains the maximum in (\ref{starting-point2}) if and only if
\begin{enumerate}
  \item[i)] Power:
  \begin{equation} \label{c-star-power}
  \int_{-\pi}^{\pi}|C^{\star}(e^{i\theta})|^2 \frac{d\theta}{2\pi}= P;
  \end{equation}
  \item[ii)] Output spectrum:
  \begin{equation} \label{c-star-output-spectrum}
  \eta:=\mathop{essinf}\limits_{\theta\in[-\pi,\pi)}|C^{\star}(e^{i\theta})+H(e^{i \theta})|^2 > 0;
  \end{equation}
  \item[iii)] Strong orthogonality: For some $0<\lambda\le\eta$
  \begin{equation} \label{c-star-strong-orthogonality}
  \frac{\lambda}{C^{\star}(e^{i \theta})+H(e^{i \theta})}-C^{\star}(e^{-i \theta})
  \end{equation}
  is causal.
\end{enumerate}
\end{thm}

The remainder of the paper is organized as follows. In Section~\ref{armakgfcmp}, we review relevant results from complex analysis and the theory of Hardy spaces as mathematical preliminaries that will be used in our proofs. Section~\ref{main-results} contains the main results of this paper, which can roughly summarized below:
\begin{itemize}
\item We prove in Section~\ref{unique-solution-section} that unless the noise $\{Z_n\}$ is white, the optimal solution to the optimization problem (\ref{starting-point2}) is unique; see Theorem~\ref{unique-output-input}.
\item In Section~\ref{recursive-procedure}, we propose an algorithm to recursively compute the optimal solution, which is guaranteed to converge to the unique optimal solution in theory and features an efficient implementation for a suboptimal solution in practice; see Algorithm~\ref{sequential-algorithm}.
\item In Section~\ref{armak-section}, we will establish Theorem~\ref{theorem3}, a ``more manageable'' version of Theorem~\ref{theorem2} and a natural extension to Theorem~\ref{arma1-theorem} combined, and derive a relatively more explicit formula for the ARMA($k$) Gaussian feedback capacity as a simple function evaluated at a solution to a system of equations, which is amenable to numerical computation for the cases $k=1, 2$ and possibly beyond.
\end{itemize}
Several examples are given in Section~\ref{examples-section}. More specifically, Example~\ref{arma1-example} details the fact that Theorem~\ref{theorem3} naturally extends Theorem~\ref{arma1-theorem}, and Example~\ref{arma2-example} use Theorem~\ref{theorem3} to numerically compute the feedback capacity of ARMA($k$) Gaussian channels. Focusing on the application of Algorithm~\ref{sequential-algorithm} to ARMA($k$) Gaussian channels, we discuss its efficient implementation and numerically compute lower bounds on the feedback capacity of ARMA($3$) Gaussian channels.

\section{Mathematical Preliminaries} \label{armakgfcmp}

In this section, we review a number of important theorems in complex analysis and the theory of Hardy spaces, which will be used in our proofs and may not be stated in the most general form.

Let $\mathbb{D}$ denote the open unit disk on the complex plane $\mathbb{C}$, that is,
$$
\mathbb{D}=\{z \in \mathbb{C}: |z| < 1\},
$$
and let $\partial \mathbb{D}$ and $\overline{\mathbb{D}}$ denote its boundary and closure, respectively, that is,
$$
\partial \mathbb{D}=\{z \in \mathbb{C}: |z| = 1\}, \quad \overline{\mathbb{D}}=\{z \in \mathbb{C}: |z| \leq 1\}.
$$

We first review two fundamental theorems in complex analysis, which are relatively better-known yet still included for self-containedness.

The following theorem gives the classical Cauchy's integral formula for an analytic function on $\overline{\mathbb{D}}$.
\begin{thm}[Cauchy's integral formula]
Let $U$ be an open subset of the complex plane $\mathbb{C}$ which contains $\overline{\mathbb{D}}$, and let $f: U \to \mathbb{C}$ be an analytic function. Then for any $n \geq 0$ and any $z_0 \in \mathbb{D}$, we have
$$
\oint_{\partial \mathbb{D}}\frac{f(z)}{(z-z_0)^{n+1}}\frac{dz}{2\pi i}=\frac{f^{(n)}(z_0)}{n!},
$$
where the contour integral is taken counter-clockwise, and the superscript $(n)$ denotes the $n$-th order complex derivative.
\end{thm}

The Cauchy integral formula can be used to establish the following Jensen's formula.
\begin{thm}[Jensen's formula]
Let $U$ be an open subset of the complex plane $\mathbb{C}$ which contains $\overline{\mathbb{D}}$. Let $f: U \to \mathbb{C}$ be an analytic function, and let $z_1, z_2, \dots, z_n$ denote the zeros of $f$ in $\mathbb{D}$ repeated according to multiplicity. Suppose that $f(0) \neq 0$. Then, we have
$$
\log |f(0)|=\sum _{k=1}^{n}\log \left({|z_{k}|}\right)+{\frac {1}{2\pi }}\int _{0}^{2\pi }\log |f(e^{i\theta })|\,d\theta.
$$
\end{thm}

Next, we will review some basic notions, terminology and needed results from the theory of Hardy spaces.

Let $1 \leq p < \infty$ and let $f(z)$ be an analytic function on $\mathbb{D}$. The function $f(z)$ is said to be of class $\mathcal{H}_p=\mathcal{H}_p(\mathbb{D})$ if
$$
\|f\|_{\mathcal{H}_p} \triangleq \sup_{0 < r < 1} \left(\int_{-\pi}^{\pi}|f(re^{i\theta})|^p\frac{d\theta}{2\pi}\right)^{1/p} < \infty.
$$
It is well known that by taking the pointwise radial limit, any $f(z) \in \mathcal{H}_p$ can be extended to a function $f(e^{i \theta}) \in \mathcal{L}_p=\mathcal{L}_p(\partial \mathbb{D})$, where
$$
\mathcal{L}_p(\partial \mathbb{D}) \triangleq \left\{f(e^{i \theta}): \left(\int_{-\pi}^{\pi}|f(e^{i\theta})|^p\frac{d\theta}{2\pi}\right)^{1/p} < \infty \right\}.
$$
When there is no risk of confusion, we will follow the usual convention and identify $f(z)$ and $f(e^{i \theta})$, which we may oftentimes simply denote by $f$. Then, $\mathcal{H}_p$ can be viewed as a closed vector subspace of $\mathcal{L}_p$.

For any $f \in \mathcal{H}_p$, we say that $f$ is \emph{causal} (or \emph{strictly causal}) if its Fourier coefficients $c_n$ is equal to $0$ for all $n < 0$ (or $n \leq 0$), where
$$
c_n=\int_{-\pi}^{\pi}f(e^{i\theta})e^{-in\theta}\frac{d\theta}{2\pi}=0,\quad n=0, \pm 1, \pm2, \dots.
$$
It is well known that $\mathcal{H}_p$ is precisely the subset of causal functions in $\mathcal{L}_p$. For a quick example, we note that $\mathcal{H}_2$, represented by infinite sequences indexed by $\mathbb{N} \cup \{0\}$ as
$$
\mathcal{H}_2 = \left\{\sum_{n=0}^{\infty} a_n e^{i n \theta}: \sum_{n=0}^{\infty} a_n^2 < \infty \right\},
$$
sits naturally inside the space $\mathcal{L}_2$, which can be represented by bi-infinite sequences indexed by $\mathbb{Z}$ as
$$
\mathcal{L}_2 = \left\{\sum_{n=-\infty}^{\infty} a_n e^{i n \theta}: \sum_{n=-\infty}^{\infty} a_n^2 < \infty \right\}.
$$

Now, we recall the \emph{inner-outer decomposition theorem} in the theory of Hardy spaces.
\begin{thm}[Theorem 2.8 in \cite{du70}]\label{iod}
Every function $f(z) \not \equiv 0$ in $\mathcal{H}_p$ has a unique factorization of the form $f(z)=B(z)S(z)F(z)$, where
\begin{itemize}
\item $B(z)$ is a {\bf \em Blaschke product} taking the following form:
\begin{equation} \label{Blaschke-product}
B(z)=z^m \prod_{n} \frac{|z_n|}{z_n}\frac{z_n-z}{1-\bar{z}_n z}=z^m \prod_{n} |z_n|\frac{1-z_n^{-1}z}{1-\bar{z}_n z},
\end{equation}
where $m$ is a nonnegative integer and $\{z_n\}$ is the set of all the zeros of $f(z)$ in $\mathbb{D}$,

\item $S(z)$ is a {\bf \em singular inner function}, which can be represented by the following Poisson-Stieltjes integral:
\begin{equation} \label{singular-inner-function}
S(z)=\exp\left\{-\int_{-\pi}^{\pi}\frac{e^{i\theta}+z}{e^{i\theta}-z}d\mu(\theta)\right\},
\end{equation}
where $\mu(\theta)$ is a bounded nondecreasing singular function with $\mu'(t)=0$ a.e.,

\item $F(z)$ is an {\bf \em outer function} taking the following form:
\begin{equation} \label{outer-function}
F(z)=e^{i \gamma} \exp\left\{\int_{-\pi}^{\pi}\frac{e^{i\theta}+z}{e^{i\theta}-z}\log|f(e^{i\theta})|\frac{d\theta}{2\pi}\right\},
\end{equation}
where $\gamma$ is a real constant.
\end{itemize}
\end{thm}

\begin{rem} \label{ic}
Note that it can be shown that $B(z)$ as in (\ref{Blaschke-product}) is analytic on $\mathbb{D}$ with the same set of zeros as $f(z)$, and $S(z)$ and $F(z)$ are also analytic without any zeros in $\mathbb{D}$. Furthermore, it is well known (see, e.g., Page $84$ of~\cite{koosis80}) that $S(z) \equiv 1$ if and only if
$$
\int_{-\pi}^{\pi}\log|B(re^{i\theta}) S(re^{i\theta})|d\theta\to0\quad\mbox{as }r\to1.
$$
\end{rem}

Roughly speaking, the following theorem says that a function in $\mathcal{H}_p$ is uniquely determined by its boundary values on any set of positive measure.
\begin{thm}[Theorem $2.2$ in~\cite{du70}] \label{unique-function}
Let $f(e^{i \theta}) \in\mathcal{H}_p$ be not identically $0$. Then $\{e^{i\theta}|f(e^{i\theta})=0\}$ has measure $0$ (with respect to the Lebesgue measure on $\partial \mathbb{D}$). Furthermore, if $f(e^{i \theta}), g(e^{i \theta}) \in\mathcal{H}_p$ and $f(e^{i\theta})=g(e^{i\theta})$ for all $\theta$ in a positive measure subset $T \subset [-\pi,\pi)$, then $f(e^{i \theta})=g(e^{i \theta})$ almost everywhere.
\end{thm}

\section{Main Results} \label{main-results}

\subsection{Uniqueness of Optimal $C(e^{i \theta})$} \label{unique-solution-section}

Recall that $C(e^{i \theta})$ is defined as in (\ref{c-def}), and we say $C(e^{i \theta})$ is an optimal solution if it solve the optimization problem (\ref{starting-point2}), namely, it satisfies (\ref{c-power}) and achieves the maximum in (\ref{starting-point2}). In this section, we will establish the uniqueness of optimal $C(e^{i \theta})$.

We will first need the following lemma.
\begin{lem} \label{less-than-one}
Let $C^{\star}(e^{i \theta})$ be an optimal solution to (\ref{starting-point2}). Then, for any $C(e^{i\theta})$ satisfying (\ref{c-power}), we have
$$
\int_{-\pi}^{\pi} \frac{|C(e^{i\theta})+H_Z(e^{i\theta})|^2}{|C^{\star}(e^{i\theta})+H_Z(e^{i\theta})|^2} \frac{d\theta}{2\pi} \leq 1.
$$
\end{lem}

\begin{proof}
Note that
\begin{align*}
\int_{-\pi}^{\pi} \frac{|C+H_Z|^2}{|C^{\star}+H_Z|^2} \frac{d\theta}{2\pi} & =\int_{-\pi}^{\pi} \frac{|C^{\star}+H_Z+C-C^{\star}|^2}{|C^{\star}+H_Z|^2} \frac{d \theta}{2\pi}\\
&\stackrel{(a)}{=}\int_{-\pi}^{\pi} \frac{|C^{\star}+H_Z|^2+|C-C^{\star}|^2+2(\overline{C^{\star}}+\overline{H_Z})(C-C^{\star})}{|C^{\star}+H_Z|^2} \frac{d\theta}{2\pi}\\
&= 1 + \int_{-\pi}^{\pi} \frac{|C-C^{\star}|^2}{|C^{\star}+H_Z|^2} d \theta+ 2 \int_{-\pi}^{\pi} \frac{C}{C^{\star}+H_Z} \frac{d \theta}{2 \pi}-2 \int_{-\pi}^{\pi} \frac{C^{\star}}{C^{\star}+H_Z} \frac{d \theta}{2\pi},
\end{align*}
where in deriving (a) we have used the easily verifiable fact that
$$
\int_{-\pi}^{\pi} \frac{(\overline{C^{\star}}+\overline{H_Z})(C-C^{\star})}{|C^{\star}+H_Z|^2} \frac{d\theta}{2\pi} = \int_{-\pi}^{\pi} \frac{(C^{\star}+H_Z)(\overline{C}-\overline{C^{\star}})}{|C^{\star}+H_Z|^2} \frac{d\theta}{2\pi}.
$$
Moreover, by (\ref{c-star-strong-orthogonality}), we have for almost all $\theta$,
$$
|C^*+H_Z|^2 \geq \lambda,
$$
and
$$
\int_{-\pi}^{\pi} \left(\frac{1}{C^{\star}+H_Z}- \frac{\overline{C^{\star}}}{\lambda} \right) C^{\star} d\theta=0, \quad \int_{-\pi}^{\pi} \left(\frac{1}{C^{\star}+H_Z}- \frac{\overline{C^{\star}}}{\lambda} \right) C d\theta=0.
$$
It then follows that for any $C(e^{i \theta})$ satisfying (\ref{c-power}),
\begin{align*}
\int_{-\pi}^{\pi} \frac{|C+H_Z|^2}{|C^{\star}+H_Z|^2} \frac{d\theta}{2\pi} & \stackrel{(b)}{\leq} 1 + \frac{1}{\lambda} \int_{-\pi}^{\pi} |C-C^{\star}|^2 \frac{d \theta}{2\pi}- \frac{2 P}{\lambda}+\frac{2}{\lambda} \int_{-\pi}^{\pi} C \overline{C^{\star}} \frac{d \theta}{2\pi}\\
&= 1 + \frac{1}{\lambda} \int_{-\pi}^{\pi} |C|^2 \frac{d\theta}{2\pi}+\frac{1}{\lambda} \int_{-\pi}^{\pi} |C^{\star}|^2 \frac{d\theta}{2\pi}-\frac{2}{\lambda} \int_{-\pi}^{\pi} C \overline{C^{\star}} \frac{d \theta}{2\pi}- \frac{2 P}{\lambda}+\frac{2}{\lambda} \int_{-\pi}^{\pi} C \overline{C^{\star}} \frac{d \theta}{2\pi}\\
&\leq 1,
\end{align*}
where we have used (\ref{c-star-power}) in deriving (b).
\end{proof}

The following theorem first shows that all optimal $C(e^{i \theta})$ give rise to the same $S_Y(e^{i \theta})$, the corresponding channel output PSD, and then establishes the uniqueness of optimal $C(e^{i \theta})$ when the channel noise is not white.
\begin{thm} \label{unique-output-input}
a) For any two optimal $C^{\star}(e^{i \theta})$ and $C^{\star \star}(e^{i \theta})$, we have, almost everywhere,
$$
S_Y^{\star}(e^{i \theta})=S_Y^{\star \star}(e^{i \theta}).
$$
b) Suppose that $\{Z_n\}$ is not white, that is, $S_Z(e^{i \theta})$ is not a constant function. Then, for any two optimal $C^{\star}(e^{i \theta})$ and $C^{\star \star}(e^{i \theta})$, we have, almost everywhere,
$$
C^{\star}(e^{i \theta})=C^{\star \star}(e^{i \theta}).
$$
\end{thm}

\begin{proof}
a) Using the well-known fact that for any $x > 0$,
\begin{equation} \label{log-minus-1}
\log x \leq x-1,
\end{equation}
we deduce that for all $\theta$,
\begin{equation} \label{log-minus-2}
\log\frac{S_Y^{\star \star}(e^{i \theta})}{S_Y^{\star}(e^{i \theta})}\le\frac{S_Y^{\star \star}(e^{i \theta})}{S_Y^{\star}(e^{i \theta})}-1,
\end{equation}
and thereby
\begin{equation}  \label{equal-in-the-end}
\hspace{-1cm} \int_{-\pi}^{\pi} \log S_Y^{\star \star}(e^{i \theta}) \frac{d \theta}{2\pi} \le \int_{-\pi}^{\pi} \left(\frac{S_Y^{\star \star}(e^{i \theta})}{S_Y^{\star}(e^{i \theta})}-1+\log S_Y^{\star}(e^{i \theta}) \right) \frac{d \theta}{2\pi} \stackrel{(a)}{\le} \int_{-\pi}^{\pi} \log S_Y^{\star}(e^{i \theta}) \frac{d \theta}{2\pi} \stackrel{(b)}{=} \int_{-\pi}^{\pi} \log S_Y^{\star \star}(e^{i \theta}) \frac{d \theta}{2\pi},
\end{equation}
where (a) follows from Lemma~\ref{less-than-one} and (b) follows from the fact the optimal solutions $C^{\star}(e^{i \theta})$ and $C^{\star \star}(e^{i \theta})$ give rise to the same optimal value. It then follows that the first inequality in (\ref{equal-in-the-end}) is in fact an equality, or equivalently,
$$
\int_{-\pi}^{\pi} \log\frac{S_Y^{\star \star}(e^{i \theta})}{S_Y^{\star}(e^{i \theta})} \frac{d\theta}{2\pi} = \int_{-\pi}^{\pi} \left(\frac{S_Y^{\star \star}(e^{i \theta})}{S_Y^{\star}(e^{i \theta})}-1 \right) \frac{d \theta}{2\pi},
$$
which, together with (\ref{log-minus-2}), immediately implies that almost everywhere,
$$
\log\frac{S_Y^{\star \star}(e^{i \theta})}{S_Y^{\star}(e^{i \theta})} = \frac{S_Y^{\star \star}(e^{i \theta})}{S_Y^{\star}(e^{i \theta})}-1.
$$
Now, using the fact that $\log x = x-1$ if and only if $x=1$, we deduce that for almost all $\theta$
$$
\frac{S_Y^{\star \star}(e^{i \theta})}{S_Y^{\star}(e^{i \theta})}=1,
$$
which immediately implies a), as desired.

b) We first consider the optimal solution $C^{\star}$, which satisfies i), ii) and iii) in Theorem~\ref{theorem1-reformulated}, which can be alternatively stated below:
\begin{itemize}
\item
\begin{equation}\label{c-sstar-power}
  \int_{-\pi}^{\pi}|C^{\star}(e^{i\theta})|^2 \frac{d\theta}{2\pi}= P;
\end{equation}
\item For some $\lambda^{\star} > 0$
\begin{equation}\label{c-sstar-strong-orthogonality}
  \frac{\lambda^{\star}}{C^{\star}(e^{i \theta})+H_Z(e^{i \theta})}-\overline{C^{\star}(e^{i \theta})}
\end{equation}
is causal;
\item For almost all $\theta\in[-\pi,\pi)$,
\begin{equation}\label{c-sstar-output-spectrum}
\lambda^{\star} \le |C^{\star}(e^{i\theta})+H_Z(e^{i \theta})|^2,
\end{equation}
where $\lambda^{\star}$ is as in (\ref{c-sstar-strong-orthogonality}).
\end{itemize}
From (\ref{c-sstar-strong-orthogonality}), straightforward computations yield that
\begin{equation} \label{first-equation}
\int_{-\pi}^{\pi}\left(\frac{C^{\star}(e^{i \theta})(\overline{C^{\star}(e^{i \theta})}+\overline{H_Z(e^{i \theta})})}{|C^{\star}(e^{i \theta})+H_Z(e^{i \theta})|^2}-\frac{1}{\lambda^{\star}}|C^{\star}(e^{i \theta})|^2\right)\frac{d\theta}{2\pi}=0,
\end{equation}
\begin{equation} \label{second-equation}
\int_{-\pi}^{\pi}\left(\frac{C^{\star\star}(e^{i \theta})(\overline{C^{\star}(e^{i \theta})}+\overline{H_Z(e^{i \theta})})}{|C^{\star}(e^{i \theta})+H_Z(e^{i \theta})|^2}-\frac{1}{\lambda^{\star}}C^{\star\star}(e^{i \theta})\overline{C^{\star}(e^{i \theta})} \right)\frac{d\theta}{2\pi}=0.
\end{equation}
Now, we consider the optimal solution $C^{\star \star}$, which similarly satisfies:
\begin{itemize}
\item
\begin{equation}\label{c-dstar-power}
  \int_{-\pi}^{\pi}|C^{\star \star}(e^{i\theta})|^2 \frac{d\theta}{2\pi}= P;
\end{equation}
\item For some $\lambda^{\star \star} > 0$
\begin{equation}\label{c-dstar-strong-orthogonality}
  \frac{\lambda^{\star \star}}{C^{\star \star}(e^{i \theta})+H_Z(e^{i \theta})}-\overline{C^{\star \star}(e^{i \theta})}
\end{equation}
is causal;
\item For almost all $\theta\in[-\pi,\pi)$,
\begin{equation}\label{c-dstar-output-spectrum}
\lambda^{\star \star} \le |C^{\star \star}(e^{i\theta})+H_Z(e^{i \theta})|^2,
\end{equation}
where $\lambda^{\star \star}$ is as in (\ref{c-sstar-strong-orthogonality}).
\end{itemize}
And parallel to (\ref{first-equation}) and (\ref{second-equation}), we have
\begin{equation} \label{third-equation}
\int_{-\pi}^{\pi}\left(\frac{C^{\star\star}(e^{i \theta})(\overline{C^{\star\star}(e^{i \theta})}+\overline{H_Z(e^{i \theta})})}{|C^{\star\star}(e^{i \theta})+H_Z(e^{i \theta})|^2}-\frac{1}{\lambda^{\star \star}}|C^{\star\star}(e^{i \theta})|^2\right)\frac{d\theta}{2\pi}=0,
\end{equation}
\begin{equation} \label{fourth-equation}
\int_{-\pi}^{\pi}\left(\frac{C^{\star}(e^{i \theta})(\overline{C^{\star\star}(e^{i \theta})}+\overline{H_Z(e^{i \theta})})}{|C^{\star\star}(e^{i \theta})+H_Z(e^{i \theta})|^2}-\frac{1}{\lambda^{\star \star}}C^{\star}(e^{i \theta})\overline{C^{\star\star}(e^{i \theta})} \right)\frac{d\theta}{2\pi}=0.
\end{equation}
Note that, by a), we have almost everywhere,
\begin{equation} \label{sstar-dstar}
|C^{\star}(e^{i \theta})+H_Z(e^{i \theta})|^2=|C^{\star\star}(e^{i \theta})+H_Z(e^{i \theta})|^2.
\end{equation}
Now, using (\ref{c-sstar-power}), (\ref{c-dstar-power}) and (\ref{sstar-dstar}), we deduce that (\ref{first-equation})-(\ref{second-equation})+(\ref{third-equation})-(\ref{fourth-equation}) can be simplified as
\begin{equation}\label{eq1}
\int_{-\pi}^{\pi}\left(\frac{|C^{\star\star}(e^{i \theta})-C^{\star}(e^{i \theta})|^2}{|C^{\star}(e^{i \theta})+H_Z(e^{i \theta})|^2}-\frac{1}{2}\left(\frac{1}{\lambda^{\star}}+\frac{1}{\lambda^{\star \star}}\right) |C^{\star\star}(e^{i \theta})-C^{\star}(e^{i \theta})|^2\right)\frac{d\theta}{2\pi}=0,
\end{equation}
or equivalently,
\begin{equation} \label{eq1-equivalent}
\int_{-\pi}^{\pi} \left(\frac{1}{\lambda^{\star}}+\frac{1}{\lambda^{\star \star}}-\frac{2}{|C^{\star}(e^{i \theta})+H_Z(e^{i \theta})|^2} \right)|C^{\star\star}(e^{i \theta})-C^{\star}(e^{i \theta})|^2 \frac{d\theta}{2\pi}=0.
\end{equation}
Note that, by (\ref{c-sstar-output-spectrum}), (\ref{c-dstar-output-spectrum}) and (\ref{sstar-dstar}), we have, for almost all $\theta$,
\begin{equation} \label{one-over-lambda}
\frac{1}{\lambda^{\star}}\ge\frac{1}{|C^{\star}(e^{i \theta})+H_Z(e^{i \theta})|^2}, \quad \frac{1}{\lambda^{\star \star}}\ge\frac{1}{|C^{\star \star}(e^{i \theta})+H_Z(e^{i \theta})|^2},
\end{equation}
which means the integrand in (\ref{eq1-equivalent}) is non-negative, and thereby must be $0$, that is,
\begin{equation} \label{integrand-zero}
\left(\frac{1}{\lambda^{\star}}+\frac{1}{\lambda^{\star \star}}-\frac{2}{|C^{\star}(e^{i \theta})+H_Z(e^{i \theta})|^2}\right)|C^{\star\star}(e^{i \theta})-C^{\star}(e^{i \theta})|^2=0
\end{equation}
for almost all $\theta\in[-\pi,\pi)$.

We now claim that there exists a positive measure set $T \subset \partial \mathbb{D}$ such that on $T$
\begin{equation} \label{set-T}
\frac{1}{\lambda^{\star}}+\frac{1}{\lambda^{\star \star}}>\frac{2}{|C^{\star}(e^{i \theta})+H_Z(e^{i \theta})|^2}.
\end{equation}
To see this, by way of contradiction, we suppose the opposite is true, that is, almost everywhere,
$$
\frac{1}{\lambda^{\star}}+\frac{1}{\lambda^{\star \star}} \leq \frac{2}{|C^{\star}(e^{i \theta})+H_Z(e^{i \theta})|^2},
$$
which, together with (\ref{one-over-lambda}), immediately implies that almost everywhere
$$
\frac{1}{\lambda^{\star}}=\frac{1}{\lambda^{\star \star}}=\frac{1}{|C^{\star}(e^{i \theta})+H_Z(e^{i \theta})|^2}.
$$
Some straightforward computations employing this yield
\begin{align*}
\frac{\lambda^{\star}}{C^{\star}(e^{i \theta})+H_Z(e^{i \theta})}-\overline{C^{\star}(e^{i \theta})} & = \frac{\lambda^{\star}(\overline{C^{\star}(e^{i \theta})}+\overline{H_Z(e^{i \theta})})}{|C^{\star}(e^{i \theta})+H(e^{i \theta})|^2}-\overline{C^{\star}(e^{i \theta})}\\
& = \overline{C^{\star}(e^{i \theta})}+\overline{H_Z(e^{i \theta})}-\overline{C^{\star}(e^{i \theta})}\\
& = \overline{H_Z(e^{i \theta})},
\end{align*}
which, together with (\ref{c-sstar-strong-orthogonality}), immediately implies that $\overline{H_Z}(e^{i \theta})$ is causal. Since $H_Z(e^{i \theta})$ is causal, we deduce that $H_Z(e^{i \theta})$ is a constant, and thereby $S_Z(e^{i \theta})$ is also a constant, a contradiction to the assumption that $\{Z_n\}$ is not white.

Now, with the claim in (\ref{set-T}), we infer from (\ref{integrand-zero}) that on the positive measure set $T \subset \partial \mathbb{D}$,
$$
C^{\star}(e^{i \theta}) = C^{\star \star}(e^{i \theta}),
$$
which, by Theorem~\ref{unique-function}, immediately implies b).
\end{proof}

\subsection{Computation of Optimal $C(e^{i \theta})$} \label{recursive-procedure}

Assuming $\{Z_n\}$ is not white, we give in this section a recursive algorithm to compute the unique optimal solution $C(e^{i \theta})$.

We will first consider the the following optimization problem and establish the uniqueness of its optimal solution:
\begin{align}
&\mbox{maximize} & &\hspace{-2cm} \int_{-\pi}^{\pi}\frac{|C(e^{i\theta})+H_Z(e^{i\theta})|^2}{|C^{\star}(e^{i\theta})+H_Z(e^{i\theta})|^2} \frac{d\theta}{2\pi} \nonumber \\
&\mbox{subject to} & &\hspace{-2cm} \int_{-\pi}^{\pi}|C(e^{i\theta})|^2 \frac{d\theta}{2\pi}\le P, \label{op2}
\end{align}
where $C^{\star}(e^{i\theta})$ is the unique optimal solution to (\ref{starting-point2}).

\begin{thm} \label{another-necessary-sufficient}
A solution $C^{\star \star}(e^{i\theta})$ to (\ref{op2}) is optimal if and only if the following conditions are satisfied:
\begin{itemize}
\item[i)]
\begin{equation}\label{c-ddstar-power}
  \int_{-\pi}^{\pi}|C^{\star \star}(e^{i\theta})|^2 \frac{d\theta}{2\pi}= P;
\end{equation}
\item[ii)] For some $\lambda > 0$
\begin{equation}\label{c-ddstar-strong-orthogonality}
  \frac{\lambda (\overline{C^{\star \star}(e^{i \theta})}+\overline{H_Z(e^{i \theta})})}{|C^{\star}(e^{i \theta})+H(e^{i \theta})|^2}-\overline{C(e^{i \theta})}
\end{equation}
is causal;
\item[iii)] For almost all $\theta\in[-\pi,\pi)$,
\begin{equation}\label{c-ddstar-output-spectrum}
  \lambda \le |C^{\star \star}(e^{i\theta})+H_Z(e^{i \theta})|^2
\end{equation}
where $\lambda$ is as in (\ref{c-ddstar-strong-orthogonality}).
\end{itemize}
\end{thm}

\begin{proof}
The proof is very similar to that of Theorem~\ref{theorem1}, and thus postponed to Appendix~\ref{proof-another-necessary-sufficient}.
\end{proof}

\begin{thm} \label{unique-solution}
Assume that $\{Z_n\}$ is not white. Then the optimal solution to (\ref{op2}) is unique.
\end{thm}

\begin{proof}
Note that by Lemma~\ref{less-than-one}, we have for any $C(e^{i\theta})$ satisfying (\ref{c-power}),
$$
\int_{-\pi}^{\pi} \frac{|C(e^{i\theta})+H_Z(e^{i\theta})|^2}{|C^{\star}(e^{i\theta})+H_Z(e^{i\theta})|^2} \frac{d\theta}{2\pi} \leq 1.
$$
In other words, other than being the unique optimal solution to (\ref{starting-point2}), $C^{\star}(e^{i\theta})$ is also one of the optimal solution to (\ref{op2}). Let $C^{\star \star}(e^{i\theta})$ be another optimal solution to (\ref{op2}). Then, by Theorem~\ref{another-necessary-sufficient}, $C^{\star}(e^{i\theta})$ and $C^{\star \star}(e^{i\theta})$ satisfy (\ref{c-ddstar-power}), (\ref{c-ddstar-strong-orthogonality}) and (\ref{c-ddstar-output-spectrum}) with $\lambda^{\star}$ and $\lambda^{\star \star}$, respectively. Now, a completely parallel argument as in the proof of Theorem~\ref{unique-output-input} will yield
$$
\int_{-\pi}^{\pi}\left(\frac{C^{\star}(e^{i \theta})(\overline{C^{\star}(e^{i \theta})}+\overline{H_Z(e^{i \theta})})}{|C^{\star}(e^{i \theta})+H_Z(e^{i \theta})|^2}-\frac{1}{\lambda^{\star}}|C^{\star}(e^{i \theta})|^2\right)\frac{d\theta}{2\pi}=0,
$$
$$
\int_{-\pi}^{\pi}\left(\frac{C^{\star\star}(e^{i \theta})(\overline{C^{\star}(e^{i \theta})}+\overline{H_Z(e^{i \theta})})}{|C^{\star}(e^{i \theta})+H_Z(e^{i \theta})|^2}-\frac{1}{\lambda^{\star}}C^{\star\star}(e^{i \theta})\overline{C^{\star}(e^{i \theta})} \right)\frac{d\theta}{2\pi}=0,
$$
$$
\int_{-\pi}^{\pi}\left(\frac{C^{\star\star}(e^{i \theta})(\overline{C^{\star\star}(e^{i \theta})}+\overline{H_Z(e^{i \theta})})}{|C^{\star}(e^{i \theta})+H_Z(e^{i \theta})|^2}-\frac{1}{\lambda^{\star \star}}|C^{\star\star}(e^{i \theta})|^2\right)\frac{d\theta}{2\pi}=0,
$$
$$
\int_{-\pi}^{\pi}\left(\frac{C^{\star}(e^{i \theta})(\overline{C^{\star\star}(e^{i \theta})}+\overline{H_Z(e^{i \theta})})}{|C^{\star}(e^{i \theta})+H_Z(e^{i \theta})|^2}-\frac{1}{\lambda^{\star \star}}C^{\star}(e^{i \theta})\overline{C^{\star\star}(e^{i \theta})} \right)\frac{d\theta}{2\pi}=0,
$$
which will collectively imply
\begin{equation}\label{eq2}
\int_{-\pi}^{\pi}\left(\frac{|C^{\star\star}(e^{i \theta})-C^{\star}(e^{i \theta})|^2}{|C^{\star}(e^{i \theta})+H_Z(e^{i \theta})|^2}-\frac{1}{2}\left(\frac{1}{\lambda^{\star}}+\frac{1}{\lambda^{\star \star}}\right) |C^{\star\star}(e^{i \theta})-C^{\star}(e^{i \theta})|^2\right)\frac{d\theta}{2\pi}=0,
\end{equation}
and furthermore
\begin{equation} \label{eq2-equivalent}
\left(\frac{1}{\lambda^{\star}}+\frac{1}{\lambda^{\star \star}}-\frac{2}{|C^{\star}(e^{i \theta})+H_Z(e^{i \theta})|^2}\right)|C^{\star\star}(e^{i \theta})-C^{\star}(e^{i \theta})|^2=0.
\end{equation}
for almost all $\theta \in [-\pi,\pi)$. The remainder of the proof then uses exactly the same argument as in the proof of Theorem \ref{unique-output-input} to establish
$$
C^{\star}(e^{i \theta}) = C^{\star\star}(e^{i \theta})
$$
almost everywhere and thereby the uniqueness of the optimal solution to (\ref{op2}).
\end{proof}

Now, we consider the following algorithm to compute the optimal $C^{i \theta}$ via recursively solving a sequence of optimization problems:
\begin{algo}  \label{sequential-algorithm}
\begin{enumerate}
\item[1)] Arbitrarily choose $C^{(0)}(e^{i\theta}) \in \mathcal{H}_2$ satisfying
$$
\int_{-\pi}^{\pi}|C^{(0)}(e^{i\theta})|^2 \frac{d\theta}{2\pi}\le P.
$$
\item[2)] For $n=0,1,\dots$, solve the following optimization problem
\begin{align}
&\mbox{minimize} & &\hspace{-2cm} \int_{-\pi}^{\pi}\frac{|C^{(n)}(e^{i\theta})+H_Z(e^{i\theta})|^2}{|C(e^{i\theta})+H_Z(e^{i\theta})|^2} \frac{d\theta}{2\pi}\nonumber\\
&\mbox{subject to} & &\hspace{-2cm} \int_{-\pi}^{\pi}|C(e^{i\theta})|^2 \frac{d\theta}{2\pi}\le P, \label{op3}
\end{align}
and then set $C^{(n+1)}(e^{i\theta})$ to be one of the optimal solutions.
\item[3)] Set $n=n+1$ and repeat 2).
\end{enumerate}
\end{algo}
\noindent Obviously, the above recursive procedure yields a sequence of functions $\{C^{(n)}(e^{i\theta})\}$ in $\mathcal{H}_2$. The following theorem discusses the convergence behavior of this sequence.
\begin{thm} \label{iteration}
Assume that $\{Z_n\}$ is not white. If there is a pointwise convergent subsequence $\{C^{(n_k)}(e^{i \theta})\}$ such that
\begin{equation}  \label{we-missed-before}
\lim_{k \to \infty} \int_{-\pi}^{\pi} |C^{(n_k)}(e^{i\theta})|^2 \frac{d\theta}{2\pi} = \int_{-\pi}^{\pi}  |\lim_{k \to \infty} C^{(n_k)}(e^{i\theta})|^2 \frac{d\theta}{2\pi},
\end{equation}
then $\{C^{(n_k)}(e^{i \theta})\}$ must converge to $C^{\star}(e^{i\theta})$, the unique optimal solution to (\ref{starting-point2}), almost everywhere.
\end{thm}

\begin{proof}
First of all, we will show that
\begin{equation}\label{eq4}
\lim_{n\to\infty}\int_{-\pi}^{\pi} \frac{|C^{(n)}(e^{i\theta})+H_Z(e^{i\theta})|^2}{|C^{(n+1)}(e^{i\theta})+H_Z(e^{i\theta})|^2} \frac{d\theta}{2\pi}=1.
\end{equation}
Apparently, we have, for all $i=0,1,\dots$,
$$
\int_{-\pi}^{\pi}\frac{|C^{(n)}(e^{i\theta})+H_Z(e^{i\theta})|^2}{|C^{(n+1)}(e^{i\theta})+H_Z(e^{i\theta})|^2} \frac{d\theta}{2\pi} \le 1,
$$
which immediately implies that
$$
\limsup_{n \to \infty}\int_{-\pi}^{\pi} \frac{|C^{(n)}(e^{i\theta})+H_Z(e^{i\theta})|^2}{|C^{(n+1)}(e^{i\theta})+H_Z(e^{i\theta})|^2} \frac{d\theta}{2\pi} \le 1.
$$
So, to show (\ref{eq4}), we only need to prove
\begin{equation} \label{liminf-greater-than-one}
\liminf_{n\to\infty}\int_{-\pi}^{\pi} \frac{|C^{(n)}(e^{i\theta})+H_Z(e^{i\theta})|^2}{|C^{(n+1)}(e^{i\theta})+H_Z(e^{i\theta})|^2} \frac{d\theta}{2\pi} \geq 1.
\end{equation}
To show this, suppose, by way of contradiction, that
$$
\liminf_{n\to\infty}\int_{-\pi}^{\pi} \frac{|C^{(n)}(e^{i\theta})+H_Z(e^{i\theta})|^2}{|C^{(n+1)}(e^{i\theta})+H_Z(e^{i\theta})|^2} \frac{d\theta}{2\pi}<1.
$$
Then, there exist $\delta > 0$ and a subsequence $\{C^{(n_j)}(e^{i\theta})\}_{i=0}^\infty$ such that
$$\int_{-\pi}^{\pi}\frac{|C^{(n_j)}(e^{i\theta})+H_Z(e^{i\theta})|^2}{|C^{(n_{j+1})}(e^{i\theta})+H_Z(e^{i\theta})|^2} \frac{d\theta}{2\pi}\le1-\delta$$
for all $j \in \mathbb{N}$. It then follows from
{\footnotesize$$
\int_{-\pi}^{\pi}\log|C^{(n_j)}(e^{i\theta})+H_Z(e^{i\theta})|^2\frac{d\theta}{2\pi}-\int_{-\pi}^{\pi}\log|C^{(n_{j+1})}(e^{i\theta})+H_Z(e^{i\theta})|^2\frac{d\theta}{2\pi}\le\int_{-\pi}^{\pi}\frac{|C^{(n_i)}(e^{i\theta})+H_Z(e^{i\theta})|^2}{|C^{(n_{i+1})}(e^{i\theta})+H_Z(e^{i\theta})|^2} \frac{d\theta}{2\pi}-1\le-\delta,
$$}
that
$$
\lim_{j \to \infty}\int_{-\pi}^{\pi}\log|C^{(n_j)}(e^{i\theta})+H_Z(e^{i\theta})|^2\frac{d\theta}{2\pi}=\infty.
$$
But this would imply that optimal value of the optimization problem is infinity, a contradiction. And therefore we have established (\ref{liminf-greater-than-one}) and thereby (\ref{eq4}).

Now, let $C^\infty(e^{i\theta})$ denote the pointwise limit of the subsequence $\{C^{(n_k)}(e^{i\theta})\}_{k=0}^\infty$. Applying (\ref{c-sstar-output-spectrum}), (\ref{we-missed-before}) and (\ref{eq4}), we deduce that
{\small $$
\hspace{-1cm} \int_{-\pi}^{\pi} \frac{|C^\infty(e^{i\theta})+H_Z(e^{i\theta})|^2}{|C^{\star}(e^{i\theta})+H_Z(e^{i\theta})|^2} \frac{d\theta}{2\pi} = \lim_{n \to \infty} \int_{-\pi}^{\pi} \frac{|C^{(n)}(e^{i\theta})+H_Z(e^{i\theta})|^2}{|C^{\star}(e^{i\theta})+H_Z(e^{i\theta})|^2} \frac{d\theta}{2\pi} \geq \lim_{n \to \infty} \int_{-\pi}^{\pi} \frac{|C^{(n)}(e^{i\theta})+H_Z(e^{i\theta})|^2}{|C^{(n+1)}(e^{i\theta})+H_Z(e^{i\theta})|^2} \frac{d\theta}{2\pi} = 1.
$$}
On the other hand, by Lemma~\ref{less-than-one}, we have
$$
\int_{-\pi}^{\pi} \frac{|C(e^{i\theta})+H_Z(e^{i\theta})|^2}{|C^{\star}(e^{i\theta})+H_Z(e^{i\theta})|^2} \frac{d\theta}{2\pi} \le 1
$$
for any $C(e^{i\theta})$ satisfying (\ref{c-power}). Therefore,
$$
\int_{-\pi}^{\pi} \frac{|C^\infty(e^{i\theta})+H_Z(e^{i\theta})|^2}{|C^{\star}(e^{i\theta})+H_Z(e^{i\theta})|^2} \frac{d\theta}{2\pi} = 1;
$$
in other words, $C^{\infty}(e^{i\theta})$ is an optimal solution to the optimization problem (\ref{op2}). Now, by Theorem~\ref{unique-solution}, we conclude that almost everywhere
$$
C^\infty(e^{i\theta})=C^{\star}(e^{i\theta}),
$$
and thereby completing the proof of the theorem.
\end{proof}

\begin{rem} \label{global-by-local}
Roughly speaking, Theorem~\ref{iteration} says that any convergent subsequence produced by Algorithm~\ref{sequential-algorithm} will converge to the optimal solution to (\ref{starting-point2}). Algorithm~\ref{sequential-algorithm} will practically compute the Gaussian feedback capacity if the global minimum of the optimization problem (\ref{op3}) can be computed. Although this is a feasible task for certain special families of channels, we are not aware of any efficient way to solve the optimization problem in (\ref{op3}) for a general stationary Gaussian channel, which is a great impediment for implementing Algorithm~\ref{sequential-algorithm}. One effective way to circumvent this issue is to find a local minimum in lieu of the global minimum of (\ref{op3}). Obviously, with such a replacement, the performance of the algorithm is compromised in the sense that it will only produce a suboptimal solution. On the other hand though, we have observed that the recursive update in Step 2) provides an effective means to prevent the produced sequence from getting stuck at some local optimal solution locally. As a matter of fact, for many practical channels for which we know the capacity (see Section~\ref{armak-section}), the compromised algorithm appears to be quickly convergent to the true optimal solution; see Example~\ref{arma3-example}.
\end{rem}

\subsection{Optimal $C(e^{i \theta})$ for ARMA($k$) Gaussian Channels}  \label{armak-section}

In this section, we generalize Theorem~\ref{arma1-theorem} and give a more explicit characterization of the optimal solution $C^{\star}(e^{i \theta})$ for the case that $\{Z_n\}$ is an ARMA($k$) Gaussian process.

The proof of our main result in this section will use the following lemma, whose proof closely follows that of Proposition $4.2$ in~\cite{kim10} and is included for completeness.
\begin{lem}\label{causal}
Suppose that the assumptions of Theorem~\ref{theorem1-reformulated} are satisfied. If $C^{\star}$ is an optimal solution to (\ref{starting-point2}), then $\overline{C^{\star}}(C^{\star}+H_Z)$ is causal.
\end{lem}

\begin{proof}
Suppose, by way of contradiction, that $\overline{C^{\star}}(C^{\star}+H_Z)$ is not causal, then for some $n\ge 1$, we have
$$
\int_{-\pi}^{\pi}\overline{C^{\star}}(C^{\star}+H_Z)e^{in\theta}\frac{d\theta}{2\pi}=\gamma \neq 0.
$$
Let $A(e^{i\theta})=xe^{in\theta}$ with $|x|<1$. Then, for $C^{\star\star} \triangleq (1+A)(C^{\star}+H_Z)-H_Z$, one verifies that it is also strictly causal, and furthermore,
\begin{align*}
\log S_Y^{\star\star}&=\log|C^{\star\star}+H_Z|^2\\
&=\log|1+A|^2|C^{\star}+H_Z|^2=\log|1+A|^2S_Y^{\star}.
\end{align*}
By Jensen's formula, the entropy rate of $S_Y^{\star\star}$ is the same as that of $S_Y^{\star}$. On the other hand, the power of $C^{\star}$ can be computed as follows:
\begin{align*}
P^{\star\star}(x)&=\int_{-\pi}^{\pi} |C^{\star\star}|^2 \frac{d \theta}{2\pi}\\
&=\int_{-\pi}^{\pi} |C^{\star}+A(C^{\star}+H_Z)|^2 \frac{d \theta}{2\pi} \\
&=\int_{-\pi}^{\pi} |C^{\star}|^2+2\int A\overline{C^{\star}}(C^{\star}+H_Z)+\int_{-\pi}^{\pi} |A|^2|C^{\star}+H_Z|^2 \frac{d \theta}{2\pi}\\
&=P+2\gamma x+P_Yx^2,
\end{align*}
where $P_Y=\int S_Y^{\star} d \theta/2 \pi >0$. Therefore, we can choose certain $x$ such that $P^{\star\star}(x)<P$, i.e., we can achieve same information rate using less power, which is contradictory to Condition i) of Theorem~\ref{theorem1-reformulated}.
\end{proof}

We are now ready to state the main result of this section.
\begin{thm}\label{theorem3}
Suppose the noise $\{Z_i\}$ is not white with the power spectral density $S_Z(e^{i \theta})$ taking the form as in (\ref{armaknoise}). Then, the feedback capacity $C_{FB}$ can be achieved by $C(z)$ taking the following form:
\begin{equation}\label{formofarmak}
C(z)=\sum_{i=1}^{l}\sum_{j=1}^{m_i}\frac{y_{ij}z^j}{(1-x_iz)^j},
\end{equation}
where $m_i$ are positive integers for all $i=1,2,\dots,l$ and $\sum_{i=1}^{l}m_i\le k$, $x_i\in\mathbb{C}$ are all distinct and $|x_i|<1$ for all $i=1,2,\dots,l$, $y_{ij}\in\mathbb{C}$ for all $i$ and $j$. Furthermore, $C(z)$ is optimal yielding the capacity
\begin{equation} \label{armak-capacity-formula}
C_{FB}=-\log\prod_{i=1}^{l}|x_i|^{m_i}
\end{equation}
if and only if all $x_i$, $m_i$ and $y_{ij}$ satisfy the following four conditions:
\begin{enumerate}
\item[i)] Power:
$$
\left.\sum_{i=1}^{l}\sum_{j=1}^{m_i}\sum_{p=1}^{l}\sum_{q=1}^{m_p}y_{ij}y_{pq}\left(\frac{z^{j-1}}{(1-x_iz)^j}\right)^{(q-1)}\right|_{z=x_p}=P,
$$
where, as elsewhere in this paper, the parenthesized superscript means the derivative with respect to $z$;
\item[ii)] Roots: $x_1,x_2,\dots,x_l$ are the roots of the function
$$
f(z)\triangleq\sum_{i=1}^{l}\sum_{j=1}^{m_i}\frac{y_{ij}z^j}{(1-x_iz)^j}+\frac{\prod_{i=1}^{k}(1+\alpha_iz)}{\prod_{i=1}^{k}(1+\beta_iz)},
$$
that are strictly inside the unit circle, while the other roots $r_1^{-1},r_2^{-1},\dots,r_k^{-1}$ are all strictly outside the unit circle;
\item[iii)] Strong orthogonality: there exists a real number $\lambda>0$ such that for all $i=1,2,\dots,l$ and $j=1,2,\dots,m_i$,
$$h_{ij}(x_i)=\lambda y_{ij}(j-1)!,$$
where
\begin{adjustwidth}{-3em}{0em}
{\small$$
h_{ij}(z) \triangleq \sum_{p=j}^{m_i}\frac{C_{p-1}^{j-1}}{(p-1)!(m_i-p)!}\left(\frac{(z-x_i)^{m_i}}{\prod_{s=1}^{l}(z-x_s)^{m_s}}\right)^{(m_i-p)} \times \left(\frac{\prod_{s=1}^{l}(1-x_sz)^{m_s}(-x_s)^{m_s}\prod_{t=1}^{k}(1+\beta_tz)}{\prod_{t=1}^{k}(1-r_tz)}\right)^{(p-j)}; \label{formofhij}
 $$}
\end{adjustwidth}
\item[iv)] Output spectrum: For almost all $\theta\in[-\pi,\pi)$,
$$
\lambda\ge \frac{1}{S_Y^\star(e^{i\theta})}=\prod_{j=1}^{l}|x_j|^{2m_j}\left|\frac{\prod_{t=1}^{k}(1+\beta_te^{i\theta})}{\prod_{t=1}^{k}(1-r_te^{i\theta})}\right|^2.
$$

\end{enumerate}
\end{thm}

\begin{proof}
Through a similar argument as in the proof of Theorem~\ref{theorem2}, we first show that any capacity achieving $C^{\star}(z)\triangleq\sum_{k=1}^{\infty}c^{\star}_k z^k$ must take the form in (\ref{formofarmak}). To this end, we consider $\hat{S}_Y^*(e^{i\theta}) \triangleq |Q(e^{i\theta})|^2S_Y^{\star}(e^{i\theta})$, which, by straightforward computations, can be rewritten as follows:
\begin{align}
\hat{S}_Y^*(e^{i \theta}) &=|Q(e^{i \theta})|^2|C^{\star}(e^{i \theta})+H_Z(e^{i \theta})|^2 \nonumber\\
&=|Q(e^{i \theta})|^2\overline{C^{\star}(e^{i \theta})}(C^{\star}(e^{i \theta})+H_Z(e^{i \theta}))+|Q(e^{i \theta})|^2\overline{H_Z(e^{i \theta})}(C^{\star}(e^{i \theta})+H_Z(e^{i \theta})) \nonumber\\
&=|Q(e^{i \theta})|^2\overline{C^{\star}(e^{i \theta})}(C^{\star}(e^{i \theta})+H_Z(e^{i \theta}))+\overline{P(e^{i \theta})}Q(e^{i \theta})C^{\star}(e^{i \theta})+|P(e^{i \theta})|^2. \label{S-Rewritten}
\end{align}
Now, it follows from Lemma~\ref{causal}, (\ref{S-Rewritten}) and the fact that $P(z)$ and $Q(z)$ are both polynomials of degree at most $k$ that $\hat{S}_Y^*(e^{i \theta})$ must be of the following form:
$$
\hat{S}_Y^*(e^{i\theta})=s_{-k}e^{-ik\theta}+s_{-k+1}e^{-i(k-1)\theta}+\cdots.
$$
Then, by the fact that $\hat{S}_Y^*(e^{i \theta})$ is symmetric, we deduce that on $\partial \mathbb{D}$, $\hat{S}_Y^*$ can be written as
$$
\hat{S}_Y^*(e^{i\theta})=s_{-k}e^{-ik\theta}+s_{-k+1}e^{-i(k-1)\theta}+\cdots+s_{-k+1}e^{i(k-1)\theta}+s_{-k}e^{ik\theta},
$$
or alternatively, on $\mathbb{D}$,
\begin{equation} \label{S-z}
\hat{S}_Y^*(z)=s_{-k}z^{-k}+s_{-k+1}z^{-(k-1)}+\cdots+s_{-k+1}z^{(k-1)}+s_{-k}z^{k}.
\end{equation}
Note that $\hat{S}_Y^*(e^{i \theta})$ has a canonical factorization (see Page $733, 734$ of~\cite{Priestley82}), namely, it can be written as
\begin{equation} \label{S-1}
\hat{S}_Y^*(e^{i \theta})=\sigma^2R(e^{i \theta}) \overline{R(e^{i \theta})},
\end{equation}
where $\sigma$ is a positive constant and $R(z)$ is a $k$-th order stable polynomial with $R(0)=1$. Now, we consider
\begin{equation} \label{T-1}
T(z) \triangleq \frac{(C^{\star}(z)+H_Z(z))Q(z)}{\sigma R(z)}.
\end{equation}
Since $C^{\star}(z)+H_Z(z)$ is an $\mathcal{H}_2$ function and $Q(z), R(z)$ are both stable polynomials, $T(z)$ is an $\mathcal{H}_2$ function. It then follows from (\ref{S-1}) and (\ref{T-1}) that
\begin{equation}  \label{T-is-1}
T(e^{i \theta})T(e^{-i \theta})=1,
\end{equation}
which, by (\ref{outer-function}), implies that the outer function in the inner-outer decomposition of $T(z)$ is the constant function $1$. Now, by (\ref{S-z}) and (\ref{T-1}), we have
\begin{align*}
  \int_{-\pi}^{\pi}\log|T(re^{i\theta})|d\theta=&\int_{-\pi}^{\pi}\log\left|\frac{(C^{\star}(re^{i\theta})+H_Z(re^{i\theta}))Q(re^{i\theta})}{\sigma R(re^{i\theta})}\right|d\theta\\
  =&\frac{1}{2}\int_{-\pi}^{\pi}\log\frac{|C^{\star}(re^{i\theta})+H_Z(re^{i\theta})|^2|Q(re^{i\theta})|^2}{\sigma^2 |R(re^{i\theta})|^2}d\theta\\
  =&\frac{1}{2}\int_{-\pi}^{\pi}\log\frac{\hat{S}_Y^*(e^{i \theta})}{\sigma^2 |R(re^{i\theta})|^2}d\theta.
\end{align*}
It then follows from (\ref{S-z}) and the fact that $R(z)$ is a stable polynomial that
$$
\lim_{r\to1}\int_{-\pi}^{\pi}\log|T(re^{i\theta})|d\theta=\int_{-\pi}^{\pi}\log|T(e^{i\theta})|d\theta=0,
$$
which, by Remark~\ref{ic}, implies that $T(z)$ is nothing but a Blaschke product, and furthermore, $C^{\star}(z)+H_Z(z)$ must take the following form:
\begin{align}
C^{\star}(z)+H_Z(z)=\frac{\prod_{i=1}^{\infty}(1-x _i^{-1}z)R(z)}{\prod_{i=1}^{\infty}(1-\bar{x}_i z)Q(z)}\label{formofc+hz}
\end{align}
for some complex numbers $x_1, x_2,\dots$ with $|x_j|<1$ for all $j$ and $\prod_{j}|x_j|^2=1/\sigma^2$. By Condition iii) of Theorem~\ref{theorem1-reformulated},
$$\frac{1}{C^{\star}(e^{i\theta})+H_Z(e^{i\theta})}-\lambda\overline{C^{\star}(e^{i\theta})}$$
is causal, which means that
\begin{equation}  \label{jian-causal}
\frac{1}{C^{\star}(z)+H_Z(z)}-\lambda C^{\star}(z^{-1})=\frac{1-\lambda S_Y^{\star}(z)+\lambda H_Z(z^{-1})(C^{\star}(z)+H_Z(z))}{C^{\star}(z)+H_Z(z)}
\end{equation}
is analytic on $\mathbb{D}$, which, together with the fact that $C^{\star}(z)+H_Z(z)$ has the factor of $\prod_{i=1}^{\infty}(1-x_i^{-1}z)$ (for this, see (\ref{formofc+hz})), implies that $1-\lambda S_Y^{\star}(z)$ must also have the same factor. By symmetry, $1-\lambda S_Y^{\star}(z)$ must also have the factor $\prod_{i=1}^{\infty}(1-x_i^{-1}z^{-1})$, which means that all $x_i$ and $x_1^{-1}$ are zeros of $1-\lambda S_Y^{\star}(z)$. Since $1-\lambda S_Y^{\star}(z)$ is a rational spectrum with degree at most $2k$, it has at most $2k$ zeros. Therefore, we conclude that
\begin{equation}\label{formofc+hz2}
C^{\star}(z)+H_Z(z)=\frac{\prod_{i=1}^{l}(1-x_i^{-1}z)^{m_i}R(z)}{\prod_{i=1}^{l}(1-\bar{x}_iz)^{m_i}Q(z)},
\end{equation}
where all $x_i$ are distinct with $|x_i|<1$, all $m_i$ are positive integers with $\sum_{i=1}^{l}m_i\le k$.

The causality of
$$
\frac{1}{C^{\star}(e^{i \theta})+H_Z(e^{i \theta})}-\lambda\overline{C^{\star}(e^{i \theta})}
$$
implies that for any $k=1,2,\dots$,
$$
\int_{-\pi}^{\pi}\left(\frac{1}{C^{\star}(e^{i\theta})+H_Z(e^{i\theta})}-\lambda \overline{C^{\star}(e^{i\theta})}\right)e^{ik\theta}\frac{d\theta}{2\pi}=0,
$$
which, together with (\ref{formofc+hz2}), yields
$$
\int_{-\pi}^{\pi}\frac{e^{ik\theta}\prod_{i=1}^{l}(1-\bar{x}_ie^{i\theta})^{m_i}Q(e^{i\theta})}{\prod_{i=1}^{l}(1-x_i^{-1}e^{i\theta})^{m_i}R(e^{i\theta})}\frac{d\theta}{2\pi}=\lambda c^{\star}_k.
$$
Rewriting the above integral as a line integral, we have
$$
\oint_\gamma\frac{z^{k-1}\prod_{i=1}^{l}(1-\bar{x}_iz)^{m_i}Q(z)}{\prod_{i=1}^{l}(1-x_i^{-1}z)^{m_i}R(z)}\frac{dz}{2\pi i}=\lambda c^{\star}_k,
$$
\begin{equation}\label{equationofck}
\oint_\gamma\frac{z^{k-1}\prod_{i=1}^{l}(1-\bar{x}_iz)^{m_i}(-x_i)^{m_i}Q(z)}{\prod_{i=1}^{l}(z-x_i)^{m_i}R(z)}\frac{dz}{2\pi i}=\lambda c^{\star}_k,
\end{equation}
where $\gamma$ is the unit circle. Denote
$$
h(z)\triangleq\frac{\prod_{i=1}^{l}(1-\bar{x}_iz)^{m_i}(-x_i)^{m_i}Q(z)}{R(z)}.
$$
It's easy to check that $h(z)$ is an analytic function on the unit disk since $R(z)$ is stable. Via the Heaviside cover-up method, the integrand of the LHS of (\ref{equationofck}) can be decomposed as
$$
z^{k-1}\sum_{i=1}^{l}\sum_{j=1}^{m_i}\frac{\tilde{h}_{ij}(z)}{(z-x_i)^j},
$$
where $\tilde{h}_{ij}(z)=a_{ij}h(z)$ and
$$
a_{ij}=\left.\frac{1}{(m_i-j)!}\left(\frac{(z-x_i)^{m_i}}{\prod_{s=1}^{l}(z-x_s)^{m_s}}\right)^{(m_i-j)}\right|_{z=x_i}
$$
is a constant depending on $x_i$ and $m_i$. Thus $\tilde{h}_{ij}(z)$ is also an analytic function on the unit disk for all $i, j$. Applying Cauchy's integral formula, we deduce that for any $k$,
$$
\left.\sum_{i=1}^{l}\sum_{j=1}^{m_i}\frac{(\tilde{h}_{ij}(z)z^{k-1})^{(j-1)}}{(j-1)!}\right|_{z=x_i}=\lambda c^{\star}_k,
$$
or equivalently,
\begin{equation}\label{leftofck}
\left.\sum_{i=1}^{l}\sum_{j=1}^{m_i}\frac{\sum_{p=1}^{\min\{j,k\}}C_{j-1}^{p-1}a_{ij}(h(z))^{(j-p)}(z^{k-1})^{(p-1)}}{(j-1)!}\right|_{z=x_i}=\lambda c^{\star}_k.
\end{equation}
Hence, each $c^{\star}_k$ takes the following form
\begin{equation}\label{formofck}
  \sum_{i=1}^{l}\sum_{j=1}^{\min\{m_i,k\}}\tilde{y}_{ij}(k-1)\cdots(k-j+1)x_i^{k-j},
\end{equation}
where $\tilde{y}_{ij}$ is a constant independent of $k$, which immediately implies that
\begin{align}
C^{\star}(z)=&\sum_{k=1}^{\infty}c^{\star}_k z^k\nonumber\\
=&\sum_{k=1}^{\infty}\sum_{i=1}^{l}\sum_{j=1}^{\min\{m_i,k\}}\tilde{y}_{ij}(k-1)\cdots(k-j+1)x_i^{k-j}z^k\nonumber\\
=&\sum_{i=1}^{l}\sum_{j=1}^{m_i}\sum_{k=j}^{\infty} \tilde{y}_{ij}(k-1)\cdots(k-j+1)x_i^{k-j}z^k\nonumber\\
=&\sum_{i=1}^{l}\sum_{j=1}^{m_i}\frac{y_{ij}z^j}{(1-x_iz)^j},\label{formofc}
\end{align}
where $y_{ij} \triangleq \tilde{y}_{ij}/(j-1)!$. Hence, together with (\ref{formofc+hz2}),
\begin{align*}
C^{\star}(z)+H_Z(z)=&\frac{\prod_{i=1}^{l}(1-x_i^{-1}z)^{m_i}R(z)}{\prod_{i=1}^{l}(1-\bar{x}_iz)^{m_i}Q(z)}\nonumber\\
=&\sum_{i=1}^{l}\sum_{j=1}^{m_i}\frac{y_{ij}z^j}{(1-x_iz)^j}+\frac{P(z)}{Q(z)},
\end{align*}
where for the last equality, all $\bar{x}_i$ are replaced by $x_i$, which can be justified by the fact that $\{x_i\}=\{\bar{x}_i\}$, thanks to the fact that $C^{\star}(z)$ has only real-valued coefficients.

We next prove that Conditions i)-iv) are necessary and sufficient for the optimality of $C^{\star}(z)$, which, given (\ref{formofc}), readily follows from Theorem~\ref{theorem1} and some technical computations.

First of all, Condition i) follows from (\ref{formofc}) and Condition i) in Theorem~\ref{theorem1}:
\begin{align*}
&\int_{-\pi}^{\pi}|C^{\star}(e^{i\theta})|^2\frac{d\theta}{2\pi}\\
=&\int_{-\pi}^{\pi}\left|\sum_{k=1}^{l}\sum_{j=1}^{m_k}\frac{y_{kj}e^{ij\theta}}{(1-x_ke^{i\theta})^j}\right|^2\frac{d\theta}{2\pi}\\
=&\int_{-\pi}^{\pi}\sum_{k=1}^{l}\sum_{j=1}^{m_k}\sum_{p=1}^{l}\sum_{q=1}^{m_p}\frac{y_{kj}e^{ij\theta}}{(1-x_ke^{i\theta})^j}\frac{\bar{y}_{pq}e^{-iq\theta}}{(1-\bar{x_p}e^{-i\theta})^q}\frac{d\theta}{2\pi}\\
=&\int_{-\pi}^{\pi}\sum_{k=1}^{l}\sum_{j=1}^{m_k}\sum_{p=1}^{l}\sum_{q=1}^{m_p}\frac{y_{kj}\bar{y}_{pq}e^{ij\theta}}{(1-x_ke^{i\theta})^j(e^{i\theta}-\bar{x_p})^j}\frac{d\theta}{2\pi}\\
=&\oint_\gamma\sum_{k=1}^{l}\sum_{j=1}^{m_k}\sum_{p=1}^{l}\sum_{q=1}^{m_p}\frac{y_{kj}\bar{y}_{pq}z^{j-1}}{(1-x_kz)^j(z-\bar{x_p})^q}\frac{dz}{2\pi i}\\
=&\left.\sum_{k=1}^{l}\sum_{j=1}^{m_k}\sum_{p=1}^{l}\sum_{q=1}^{m_p}y_{kj}\bar{y}_{pq}\left(\frac{z^{j-1}}{(1-x_kz)^j}\right)^{(q-1)}\right|_{z=\bar{x_p}}\\
\stackrel{(a)}{=}&\left.\sum_{k=1}^{l}\sum_{j=1}^{m_k}\sum_{p=1}^{l}\sum_{q=1}^{m_p}y_{kj}y_{pq}\left(\frac{z^{j-1}}{(1-x_kz)^j}\right)^{(q-1)}\right|_{z=x_p}\\
=&P,
\end{align*}
where for (a), we have replaced $\bar{y}_{p q}$ by $y_{p q}$, which can be justified by the fact that $\{y_{p q}\}=\{\bar{y}_{p q}\}$, again due to the fact that $C^{\star}(z)$ has only real-valued coefficients.

\par Second, it follows from (\ref{formofc+hz2}) and (\ref{formofc}) that
\begin{align}
    C^{\star}(z)+H_Z(z)=&\frac{\prod_{i=1}^{l}(1-x_i^{-1}z)^{m_i}R(z)}{\prod_{i=1}^{l}(1-x_iz)^{m_i}Q(z)}\nonumber\\
    =&\sum_{i=1}^{l}\sum_{j=1}^{m_i}\frac{y_{ij}z^j}{(1-x_iz)^j}+\frac{P(z)}{Q(z)},\label{equationofck+hz}
  \end{align}
which immediately implies Condition ii).

\par Condition iii) follows from the fact that the coefficients of each $x_i^{k-j}$ at both sides of (\ref{equationofck}) are equal. More precisely, by (\ref{formofck}), the coefficient of $x_i^{k-j}$ on the right hand side is $(j-1)!(k-1)\cdots(k-j+1)\lambda y_{ij}$. On the other hand, via (\ref{leftofck}), the coefficient of $x_i^{k-j}$ on the LHS of (\ref{equationofck}) is as follows:
$$
\left.(k-1)\cdots(k-j+1)\sum_{p=j}^{m_i}\frac{C_{p-1}^{j-1}a_{ip}(h(z))^{(p-j)}}{(p-1)!}\right|_{z=x_i}.
$$
Condition iii) then immediately follows.

\par Last, Condition iv) follows from Condition iii) of Theorem~\ref{theorem1} and some technical computations.

\par Finally, noting the uniqueness of the output PSD $S^{\star}_Y$ corresponding to the optimal $C^{\star}(z)$ (Theorem~\ref{unique-output-input}) and applying Jensen's formula, we obtain
\begin{align*}
  C_{FB}=&\frac{1}{2}\int_{-\pi}^{\pi}\log S_Y^{\star}(e^{i\theta})\frac{d\theta}{2\pi}\\
  =&\frac{1}{2}\int_{-\pi}^{\pi}\log\prod_{j=1}^{l}|x_j|^{-2m_j}\left|\frac{\prod_{t=1}^{k}(1-r_te^{i\theta}))}{\prod_{t=1}^{k}(1+\beta_te^{i\theta})}\right|^2\frac{d\theta}{2\pi}\\
  =&-\log\prod_{i=1}^{l}|x_i|^{m_i}.
\end{align*}
The proof of Theorem~\ref{theorem3} is then complete.
\end{proof}

\begin{rem}
By Theorem~\ref{theorem3}, to compute the ARMA($k$) Gaussian feedback capacity, one needs to first find a solution to one of the following systems of rational equations: for some positive $m_1, m_2, \dots, m_l$ with $\sum_{j=1}^{l} m_i \leq k$,
\begin{equation}  \label{soe}
\hspace{-5mm} \left\{\begin{array}{l}
\left.\sum_{i=1}^{l}\sum_{j=1}^{m_i}\sum_{p=1}^{l}\sum_{q=1}^{m_p}y_{ij}y_{pq}\left(\frac{z^{j-1}}{(1-x_iz)^j}\right)^{(q-1)}\right|_{z=x_p}=P,\\
f(x_j)=0, \quad j=1, 2, \dots, l\\
y_{ij}h_{11}(x_1) (j-1)! =y_{11}h_{ij}(x_{i}),\quad i=1, 2, \cdots, l,\\
\qquad\qquad\qquad\qquad\qquad\qquad\qquad\quad j=1, 2,\dots, m_i,
\end{array}\right.
\end{equation}
such that $|x_i| < 1$ for all $i$ and it also satisfies Condition iv) in Theorem~\ref{theorem3} to compute the capacity with (\ref{armak-capacity-formula}).
\end{rem}

\section{Examples and Numerical Results} \label{examples-section}

In this section, we give a couple of examples and some numerical results.

\begin{exmp} \label{arma1-example}
\par When $k=1$, both $l$ and $m_l$ are necessarily $1$, and the corresponding system of equations is:
$$
\left\{\begin{array}{l}
\frac{y_{11}^2}{1-x_1^2}=P,\\
\frac{y_{11} x_1}{1-x_1^2}+\frac{1+\alpha x_1}{1+\beta x_1}=0,\\
\end{array}\right.
$$
which immediately gives rise to (\ref{arma1-equation}). An elementary analysis (see, e.g.,~\cite{kim10} or~\cite{liu16}) will show that Condition iv) of Theorem~\ref{theorem3} translates to (\ref{conditionofx}), an extra condition $x$ has to satisfy. It turns out that for this case, $x_1$ is unique, which, by (\ref{armak-capacity-formula}), yields
$$
C_{FB}=-\log|x_1|.
$$
So, Theorem~\ref{theorem3} recovers Theorem~\ref{arma1-theorem} as a special case.
\end{exmp}

\begin{exmp} \label{arma2-example}
\par When $k=2$, by Theorem~\ref{theorem3}, we have three cases to deal with:
\begin{enumerate}
\item \underline{$l=1$ and $m_1=1$}: We need to find $|x_1| < 1$, $y_{11} \neq 0$ such that
\begin{equation} \label{soe1}
\left\{\begin{array}{l}
\frac{y_{11}^2}{1-x_1^2}=P,\\
\frac{y_{11} x_1}{1-x_1^2}+\frac{(1+\alpha_1 x_1)(1+\alpha_2 x_1)}{(1+\beta_1 x_1)(1+\beta_2 x_1)}=0,\\
\end{array}\right.
\end{equation}
and for all $\theta\in[-\pi,\pi)$,
{\scriptsize$$\frac{x_1 (1+\beta_1 x_1)(1+\beta_2 x_1)(x_1^2-1)}{y_{11} (1-r_1 x_1)(1-r_2 x_1)}\ge |x_1|^2\left|\frac{(1+\beta_1 e^{i\theta})(1+\beta_2 e^{i\theta})}{(1-r_1 e^{i\theta})(1-r_2 e^{i\theta})}\right|^2,
$$}
\noindent where $r_1+r_2=x_1-x_1^{-1}-\alpha_1-\alpha_2-y_{11}$ and $r_1 r_2=\alpha_1\alpha_2 x_1^2-\beta_1 \beta_2 x_1 y_{11}$. If such $x_1$ exists, we have
$$
C_{FB}=-\log|x_1|.
$$

\item \underline{$l=1$ and $m_1=2$}: We need to find $|x_1| < 1$ and $y_{11}, y_{12} \neq 0$ such that
\begin{equation} \label{soe2}
\left\{\begin{array}{l}
  \frac{y_{11}^2}{1-x_1^2}+\frac{y_{12}^2 (1+x_1^2)}{(1-x_1^2)^2}+\frac{2y_{11}y_{12}x_1}{(1-x_1^2)^2}=P,\\
    \frac{y_{11} x_1}{1-x_1^2}+\frac{y_{12} x_1^2}{(1-x_1^2)^2}+\frac{(1+\alpha_1 x_1)(1+\alpha_2 x_1)}{(1+\beta_1 x_1)(1+\beta_2 x_1)}=0,\\
      \;\; y_{11} w(x_1)=y_{12} w^{(1)}(x_1),
  \end{array}\right.
\end{equation}
and for all $\theta\in[-\pi,\pi)$
  $$
  \frac{w(x_1)}{y_{12}}\ge |x_1|^4\left|\frac{(1+\beta_1e^{i\theta})(1+\beta_2e^{i\theta})}{(1-r_1e^{i\theta})(1-r_2e^{i\theta})}\right|^2,
  $$
where
  $$
  w(z) \triangleq \frac{x_1^2(1-x_1 z)^2(1+\beta_1z)(1+\beta_2z)}{(1-r_1z)(1-r_2z)},
  $$
and $r_1+r_2=2x_1-2x_1^{-1}-\alpha_1-\alpha_2-y_{11}$ and $r_1r_2=\alpha_1\alpha_2x_1^4-\beta_1\beta_2x_1^3y_{11}-\beta_1\beta_2x^2y_{12}$. If such $x_1, y_{11}, y_{12}$ exist, then we have
$$
C_{FB}= - \log |x_1|^2.
$$

\item \underline{$l=2$ and $m_1=1$, $m_2=1$}: We need to find distinct $|x_1|, |x_2| < 1$ and $y_{11}, y_{21} \neq 0$ such that
\begin{equation} \label{soe3}
\hspace{-9mm}\left\{\begin{array}{l}
  \frac{y_{11}^2}{1-x_1^2}+\frac{y_{21}^2}{1-x_2^2}+\frac{2y_{11}y_{21}}{1-x_1x_2}=P,\\
    \frac{y_{11} x_1}{1-x_1^2}+\frac{y_{21} x_1}{1-x_1x_2}+\frac{(1+\alpha_1x_1)(1+\alpha_2x_1)}{(1+\beta_1x_1)(1+\beta_2x_1)}=0,\\
    \frac{y_{11} x_2}{1-x_1x_2}+\frac{y_{21} x_2}{1-x_2^2}+\frac{(1+\alpha_1 x_2)(1+\alpha_2 x_2)}{(1+\beta_1 x_2)(1+\beta_2 x_2)}=0,\\
    \frac{y_{11}(1-r_1x_1)(1-r_2x_1)}{(1+\beta_1x_1)(1+\beta_2x_1)(1-x_1^2)}=-\frac{y_{21}(1-r_1x_2)(1-r_2x_2)}{(1+\beta_1x_2)(1+\beta_2x_2)(1-x_2^2)},
  \end{array}\right.
\end{equation}
and for all $\theta\in[-\pi, \pi)$,
  \begin{adjustwidth}{-1.5em}{0em}
    {\tiny$$\frac{x_1x_2(1+\beta_1x_1)(1+\beta_2x_1)(1-x_1^2)(1-x_1x_2)}{(x_1-x_2)y_{11}(1-r_1x_1)(1-r_2x_1)}\ge |x_1x_2|^2\left|\frac{(1+\beta_1e^{i\theta})(1+\beta_2e^{i\theta})}{(1-r_1e^{i\theta})(1-r_2e^{i\theta})}\right|^2$$}
  \end{adjustwidth}
where $r_1+r_2=x_1+x_2-x_1^{-1}-x_2^{-1}-\alpha_1-\alpha_2-y_{11}-y_{21}$ and $r_1r_2=\alpha_1\alpha_2x_1^2x_2^2-\beta_1\beta_2x_1^2x_2y_{21}-\beta_1\beta_2x_1x_2^2y_{11}$. If such $x_1, x_2, y_{11}, y_{21}$ exist, then we have
$$
C_{FB}=-\log|x_1 x_2|.
$$
\end{enumerate}
Complicated as they may look, the systems of equations in (\ref{soe1}), (\ref{soe2}) and (\ref{soe3}) all have finitely many solutions for generic $\alpha_1, \alpha_2, \beta_1, \beta_2$ and therefore can be numerically solved (for instance, Bertini~\cite{BHSW06}, a numerical algebraic geometry package, can be used to efficiently find their zero-dimensional roots). Below, fixing $P=1$, $\alpha_2=0.1$, and $\beta_2=0$, assuming different values for $\beta_1$, we have plotted the values of $C_{FB}$ against the values of $\alpha_1$.

\begin{figure}[htbp!]
\centering
\includegraphics [scale=0.7]{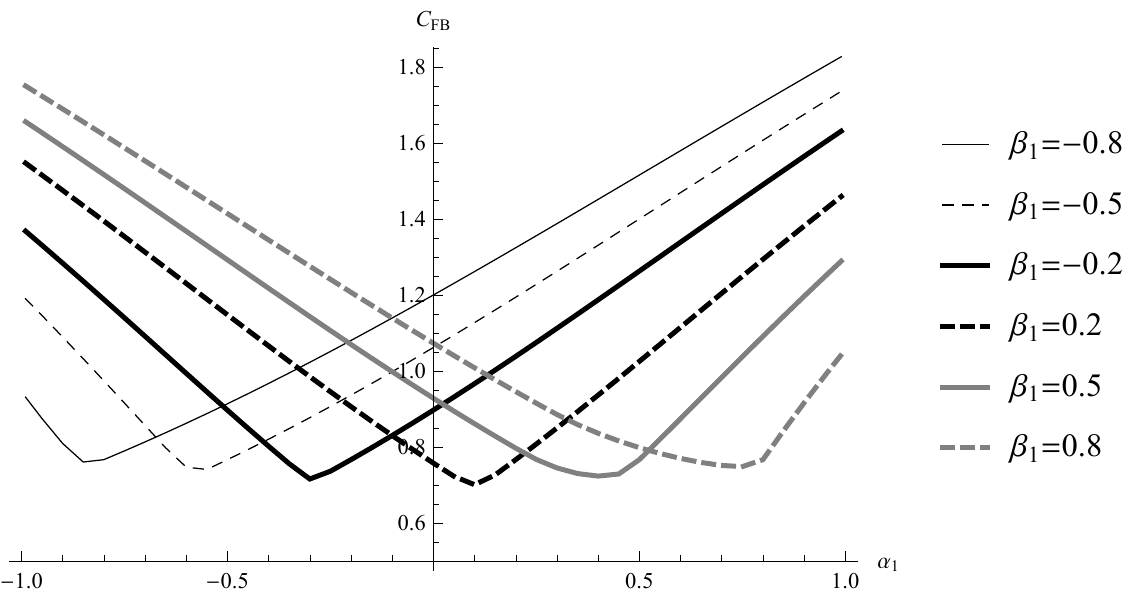}
\caption{Plot of $C_{FB}$ as a function of $\alpha_1$ when $\alpha_2=0.1$}
\end{figure}
\end{exmp}

\begin{exmp} \label{arma3-example}
As evidenced in Example~\ref{arma2-example}, solving the polynomial system in (\ref{soe}) will yield the ARMA($k$) Gaussian feedback capacity. Nevertheless, the computational complexity drastically increases as $k$ gets larger. Our observation is that with this approach, the computation can be measured in minutes (for moderate computing power) for $k=2$, but it will be measured in days for $k=3$. In this example, we demonstrate the effectiveness of Algorithm~\ref{sequential-algorithm} in terms of computing/estimating Gaussian feedback capacity. Apparently this algorithm works for much more general settings, but for the purpose of comparison, we will also focus on applying the algorithm to compute the ARMA($k$) Gaussian feedback channels.

We first discuss a couple of technical issues for the implementation of Algorithm~\ref{sequential-algorithm}.

The first issue is about the form that $C(z)$ should take for implementing the algorithm. Note that, albeit explicit, the expression as in (\ref{formofarmak}) gives different forms for different $l$ and $m_1, m_2, \dots, m_l$, which will create technical problems for Step 2), where the recursive computation of $\{C^{(n)}(e^{i \theta})\}$ is conducted. One way to circumvent this issue is to adopt the following unified form:
\begin{equation} \label{unified-form}
\frac{\sum_{n=1}^{k} \hat{y}_n e^{in\theta}} {\prod_{n=1}^{k}(1-\hat{x}_n e^{i\theta})},
\end{equation}
where $\hat{y}_n$ are complex numbers and $\hat{x}_n$ are complex numbers inside unit circle. One verifies that the above form encompasses all the possible cases in (\ref{formofarmak}).

As in Remark~\ref{global-by-local}, as there does not seem to exist an effective way to find the global minimum for (\ref{op3}), we instead update the sequence $\{C^{(n)}(e^{i \theta})\}$ by a local minimum in (\ref{op3}) via some gradient-descent like method. This, however, create another problem for choosing the initial $C^{(0)}(e^{i \theta})$; more specifically, if $C^{(0)}(e^{i\theta})$ is chosen such that $C^{(0)}(e^{i\theta})+H_Z(e^{i\theta})$ has no zeros inside the unit circle, and thereby any $C(e^{i\theta})$ ``close'' to $C^{(0)}(e^{i\theta})$, $C(e^{i\theta})+H_Z(e^{i\theta})$ will likely not have zeros inside the unit circle either. Then by Jensen's formula,
$$
\int_{-\pi}^{\pi}\frac{|C^{(0)}(e^{i\theta})+H_Z(e^{i\theta})|^2}{|C(e^{i\theta})+H_Z(e^{i\theta})|^2} \frac{d\theta}{2\pi}\ge\int_{-\pi}^{\pi}\log\frac{|C^{(0)}(e^{i\theta})+H_Z(e^{i\theta})|^2}{|C(e^{i\theta})+H_Z(e^{i\theta})|^2} \frac{d\theta}{2\pi}+1 \ge 1.
$$
Therefore, it is difficult to use a gradient-like method to find a feasible $C^{(1)}(e^{i\theta})$ such that
$$
\int_{-\pi}^{\pi}\frac{|C^{(0)}(e^{i\theta})+H_Z(e^{i\theta})|^2}{|C^{(1)}(e^{i\theta})+H_Z(e^{i\theta})|^2} \frac{d\theta}{2\pi} < 1,
$$
not to mention to find a local minimum point $C^{(1)}(e^{i\theta})$. To overcome this issue, one can further assume $C^{(0)}(e^{i \theta})$ is chosen such that $C^{(0)}(e^{i\theta})+H_Z(e^{i\theta})$ has at least one zero (denote by $s$ below) inside the unit circle, that is,
\begin{align}
C^{(0)}(e^{i\theta})+H_Z(e^{i\theta}) & = \frac{\sum_{n=1}^{k} \hat{y}_n e^{in\theta}}{\prod_{n=1}^{k}(1-\hat{x}_n e^{i\theta})}+\frac{\prod_{n=1}^{k}(1+\alpha_n e^{i\theta})}{\prod_{n=1}^{k}(1+\beta_n e^{i\theta})} \nonumber \\
&=\frac{(1-s^{-1} e^{i\theta})(1+\sum_{n=1}^{2k-1} \gamma_n e^{in\theta})}{\prod_{n=1}^{k}(1-\hat{x}_n e^{i\theta})(1+\beta_n e^{i\theta})},
\end{align}
where $|s| < 1$, $\gamma_1, \gamma_2, \dots, \gamma_{2k-1}$ are appropriately chosen complex numbers.

With these two issues addressed, Algorithm~\ref{sequential-algorithm} can be efficiently implemented to yield a lower bound (denoted by $C_{FB}^{(low)}$) on the Gaussian feedback capacity. We observe that for the ARMA($k$) channels, $k=1, 2$, the implemented algorithm actually quickly converges to the true capacity; moreover, it can also handle larger $k$'s within reasonably short time (measured in hours with moderate computing pwoer). Below, fixing $P=10$, $\alpha_1=0.3$, $\alpha_2=0.4$, $\beta_1=-0.3$, $\beta_2=0.7$, assuming different values for $\alpha_3$, we have plotted the values of $C_{FB}^{(low)}$ against the values of $\beta_3$.
\begin{figure}[htbp!]
\centering
\includegraphics [scale=0.7]{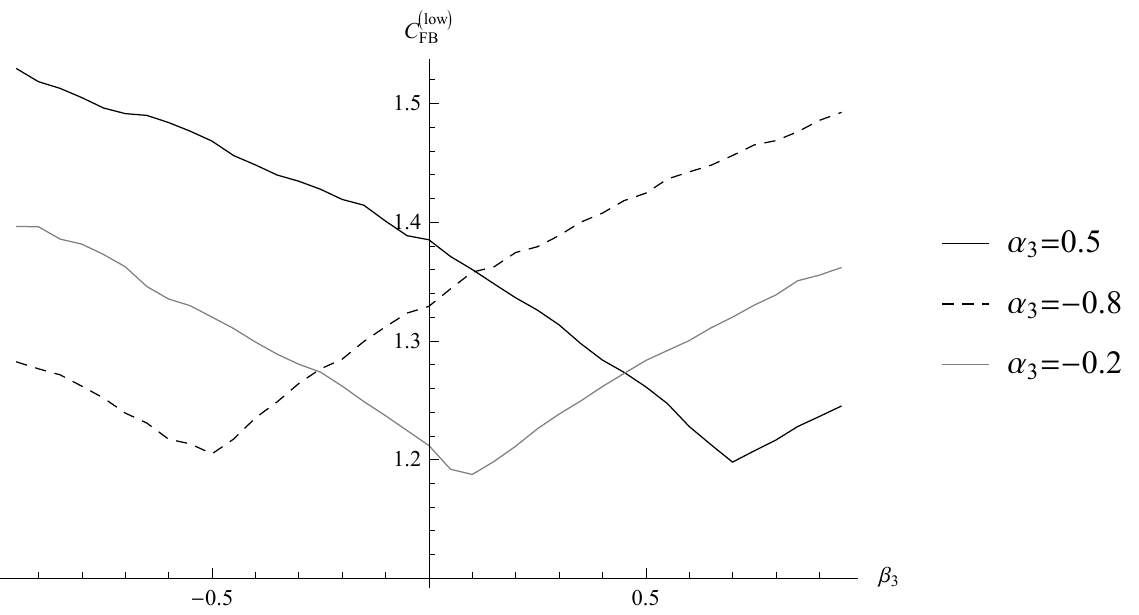}
\caption{Plot of $C_{FB}^{(low)}$ as a function of $\beta_3$}
\end{figure}
\end{exmp}

\section*{Appendices} \appendix

\section{Proof of Theorem~\ref{another-necessary-sufficient}} \label{proof-another-necessary-sufficient}

For the necessity part, we directly use the method of Lagrangian multiplier. Consider the Lagragian of (\ref{op2})
\begin{equation}
L(c,\lambda)=\int_{-\pi}^{\pi}\frac{\lambda|C(e^{i\theta})+H_Z(e^{i\theta})|^2}{|C^{\star}(e^{i\theta})+H_Z(e^{i\theta})|^2}\frac{d\theta}{2\pi}-\left(\int_{-\pi}^{\pi}|C(e^{i\theta})|^2\frac{d\theta}{2\pi}-P\right).\label{lagragian}
\end{equation}
Apparently $C^{\star\star}(e^{i \theta})$ satisfies the KKT condition, that is,
$$
\int_{-\pi}^{\pi}|C^{\star\star}(e^{i\theta})|^2 \frac{d\theta}{2\pi}= P,
$$
and for any $k=1,2,\dots$,
$$
\int_{-\pi}^{\pi}2e^{ik\theta}\left(\frac{\lambda(\overline{C^{\star\star}(e^{i \theta})}+\overline{H_Z(e^{i \theta})})}{|C^{\star}(e^{i \theta})+H(e^{i \theta})|^2}-\overline{C^{\star\star}(e^{i \theta})}\right)\frac{d\theta}{2\pi}=0,
$$
which yield (\ref{c-dstar-power}) and (\ref{c-dstar-strong-orthogonality}), respectively. Furthermore, the infinite-dimensional Hessian matrix $H$ of $L(c, \lambda)$ can be computed as
$$
H_{k,k}=\int_{-\pi}^{\pi}\frac{2\lambda}{|C^{\star}(e^{i \theta})+H_Z(e^{i \theta})|^2} \frac{d\theta}{2\pi}-2,
$$
for all feasible $k$, and
$$
H_{k,j}=\int_{-\pi}^{\pi}\frac{2\lambda e^{i|j-k|\theta}}{|C^{\star}(e^{i \theta})+H_Z(e^{i \theta})|^2} \frac{d\theta}{2\pi}
$$
for all all feasible $j \neq k$. Note that $H$ can be decomposed as $2\lambda A-2I$, where
$$
A_{k,j}=\int_{-\pi}^{\pi}\frac{2 e^{i|j-k|\theta}}{|C^{\star}(e^{i \theta})+H_Z(e^{i \theta})|^2} \frac{d\theta}{2\pi}
$$
for all feasible $j, k$. Now, at the global maximum solution $C^{\star\star}(e^{i \theta})=\sum_{j=1}^{\infty} c^{\star\star}_j e^{i j \theta}$, $H$ must satisfy: for any $n$ and any $z=(z_1, z_2,\dots, z_n) \neq \mathbf{0}$ with $\sum_{i=1}^{n} c^{\star\star}_i z_i=0$,
$$
zHz^T=\sum_{j=1}^{n}\sum_{k=1}^{n}H_{k,j}^{(n)} z_j z_k \le 0,
$$
where $H^{(n)}$ the leading principle $n \times n$ submatrix of $H$, i.e., $H^{(n)}=(H_{j, k})_{j, k=1}^n$. It then follows that at most $1$ eigenvalue of $H^{(n)}$ is positive, or equivalently, at most $1$ eigenvalue of $A^{(n)}$ is larger than $1/\lambda$, where $A^{(n)}$ is the leading principle $n \times n$ submatrix of $A$. Denote by $\lambda^{(n)}_2$ the second largest eigenvalue of $A^{(n)}$, then $\lambda^{(n)}_2\le1/\lambda$ for all $n$. It then follows from the well-known fact on the eigenvalue distribution of Toeplitz forms (see, Page 63 of~\cite{grenander58}), $\lambda^{(n)}_2$ converges to $\mathop{esssup}\limits_{\theta\in[-\pi,\pi)} |C^{\star}(e^{i \theta})+H_Z(e^{i \theta})|^{-2}$ as $n$ tends to infinity. Therefore, we conclude that
\begin{equation}\label{condition3}
\lambda\le|C^{\star}(e^{i\theta})+H(e^{i \theta})|^2
\end{equation}
for almost all $\theta \in [-\pi,\pi)$.

For the sufficiency part, we use the same idea as given in the proof in Theorem 4.1 in \cite{kim10}. More precisely, we need to prove that for any $C(e^{i\theta})$ satisfying (\ref{c-power}),
$$
\int_{-\pi}^{\pi} \frac{|C(e^{i\theta})+H_Z(e^{i\theta})|^2}{|C^{\star}(e^{i\theta})+H_Z(e^{i\theta})|^2} \frac{d\theta}{2\pi} \leq \int_{-\pi}^{\pi} \frac{|C^{\star\star}(e^{i\theta})+H_Z(e^{i\theta})|^2}{|C^{\star}(e^{i\theta})+H_Z(e^{i\theta})|^2} \frac{d\theta}{2\pi}.
$$
To see this, note that
\begin{align*}
\int_{-\pi}^{\pi} \frac{|C+H_Z|^2}{|C^{\star}+H_Z|^2} \frac{d\theta}{2\pi} =&\int_{-\pi}^{\pi} \frac{|C^{\star\star}+H_Z+C-C^{\star}|^2}{|C^{\star\star}+H_Z|^2} \frac{d \theta}{2\pi}\\
=&\int_{-\pi}^{\pi} \frac{|C^{\star\star}+H_Z|^2+|C-C^{\star\star}|^2+2(\overline{C^{\star\star}}+\bar{H}_Z)(C-C^{\star\star})}{|C^{\star}+H_Z|^2} \frac{d\theta}{2\pi}\\
=& \int_{-\pi}^{\pi} \frac{|C^{\star\star}+H_Z|^2}{|C^{\star}+H_Z|^2} \frac{d \theta}{2 \pi} + \int_{-\pi}^{\pi} \frac{|C-C^{\star\star}|^2}{|C^{\star}+H_Z|^2} \frac{d \theta}{2 \pi}+ 2 \int_{-\pi}^{\pi} \frac{C(\overline{C^{\star\star}}+\overline{H_Z})}{|C^{\star}+H_Z|^2} \frac{d \theta}{2 \pi}\\
&-2 \int_{-\pi}^{\pi} \frac{C^{\star\star}(\overline{C^{\star\star}}+\overline{H_Z})}{|C^{\star}+H_Z|^2} \frac{d \theta}{2\pi}.
\end{align*}
Note that by (\ref{c-dstar-strong-orthogonality}), we have for almost all $\theta$,
$$
|C^*+H_Z|^2 \geq \lambda,
$$
and
$$
\int_{-\pi}^{\pi} \left(\frac{(\overline{C^{\star\star}}+\overline{H_Z})}{C^{\star\star}+H_Z}- \frac{\overline{C^{\star}}}{\lambda} \right) C^{\star\star} d\theta=0, \quad \int_{-\pi}^{\pi} \left(\frac{(\overline{C^{\star\star}}+\overline{H_Z})}{C^{\star}+H_Z}- \frac{\overline{C^{\star\star}}}{\lambda} \right) C d\theta=0.
$$
It then follows that for any $C(e^{i \theta})$ satisfying (\ref{c-power}), we have
\begin{align*}
\int_{-\pi}^{\pi} \frac{|C+H_Z|^2}{|C^{\star}+H_Z|^2} \frac{d\theta}{2\pi} \leq & \int_{-\pi}^{\pi} \frac{|C^{\star\star}+H_Z|^2}{|C^{\star}+H_Z|^2} \frac{d \theta}{2 \pi} + \frac{1}{\lambda} \int_{-\pi}^{\pi} |C-C^{\star\star}|^2 \frac{d \theta}{2\pi}- \frac{2 P}{\lambda}+\frac{2}{\lambda} \int_{-\pi}^{\pi} C \overline{C^{\star\star}} \frac{d \theta}{2\pi}\\
=& \int_{-\pi}^{\pi} \frac{|C^{\star\star}+H_Z|^2}{|C^{\star}+H_Z|^2} \frac{d \theta}{2 \pi} + \frac{1}{\lambda} \int_{-\pi}^{\pi} |C|^2 \frac{d\theta}{2\pi}+\frac{1}{\lambda} \int_{-\pi}^{\pi} |C^{\star\star}|^2 \frac{d\theta}{2\pi}\\
&-\frac{2}{\lambda} \int_{-\pi}^{\pi} C \overline{C^{\star\star}} \frac{d \theta}{2\pi}- \frac{2 P}{\lambda}+\frac{2}{\lambda} \int_{-\pi}^{\pi} C \overline{C^{\star\star}} \frac{d \theta}{2\pi}\\
\leq& \int_{-\pi}^{\pi} \frac{|C^{\star\star}+H_Z|^2}{|C^{\star}+H_Z|^2} \frac{d \theta}{2 \pi}.
\end{align*}
The proof of the theorem is then complete.

\end{document}